\documentclass[a4paper,10pt,twocolumn,preprint,3p]{elsarticle}
\usepackage{calc}
\usepackage{amsmath}

\usepackage{nomencl}
\renewcommand{\nomgroup}[1]{%
 	\vspace{-0.9\parsep}%
}
\usepackage{pgffor}
\usepackage{nomencl}
\usepackage{multicol}

\usepackage{amsfonts}
\usepackage{amssymb}
\usepackage{calrsfs}
\usepackage{relsize}    

\usepackage{graphicx}
\usepackage{grffile}
\usepackage{svg}
\usepackage{subfig}
\usepackage{cleveref}

\usepackage{xspace}
\usepackage{bm}
\usepackage{placeins}
\usepackage{pdfpages}
\usepackage[font={small,it}]{caption}
\usepackage{array, tabularx}
\usepackage{url}
\usepackage{xcolor}
\usepackage[export]{adjustbox}
\usepackage{framed}

\usepackage{balance}

\newenvironment{conditions*}
  {\par\vspace{\abovedisplayskip}\noindent
   \tabularx{\columnwidth}{>{$}l<{$} @{${}:{}$}
   >{\raggedright\arraybackslash}X}} {\endtabularx\par\vspace{\belowdisplayskip}}

\usepackage[linesnumbered,ruled,vlined]{algorithm2e}

\SetCommentSty{mycommfont}

\SetAlFnt{\sffamily}

\makeatletter
\renewcommand{\algocf@Vline}[1]{
  \strut\par\nointerlineskip
  \algocf@push{\skiprule}
  \hbox{\bgroup\color{cyan}\vrule\egroup%
    \vtop{\algocf@push{\skiptext}
      \vtop{\algocf@addskiptotal #1}\bgroup\color{cyan}\Hlne\egroup}}\vskip\skiphlne
  \algocf@pop{\skiprule}
  \nointerlineskip}
\renewcommand{\algocf@Vsline}[1]{
  \strut\par\nointerlineskip
  \algocf@bblockcode%
  \algocf@push{\skiprule}
  \hbox{\bgroup\color{cyan}\vrule\egroup
    \vtop{\algocf@push{\skiptext}
      \vtop{\algocf@addskiptotal #1}}}
  \algocf@pop{\skiprule}
  \algocf@eblockcode%
}
\makeatother

\SetKwProg{Fn}{Function}{}{}
\SetKwInOut{Parameter}{Parameters}
\SetKw{KwVar}{Var}
\SetKw{KwSet}{Set}
\SetKw{KwTuple}{Tuple}
\SetKw{KwBreak}{break}
\usepackage{amsthm}
\newtheorem{thm}{Theorem}

\usepackage{xargs}

\newcommand\vmath[1]{\ensuremath{#1}\xspace}

\newcommand\vDc{data consumer} 
\newcommand\vDp{data producer}	

\newcommand\visibleComment[1]{#1}

\newcommand\vcomment[1]{}

\newcommand\vUser{\vmath{n}}
\newcommand\vUserSet{\vmath{N}}
\newcommand\vTime{\vmath{t}}

\newcommand\vSensorValue{\vmath{s}}
\newcommand\vSensorValueSet{\vmath{S}}

\newcommand\vSensorValueWithIndeces[2]{\vmath{\vSensorValue_{#1, #2}}}
\newcommand\vSensorValueIndexed{\vmath{\vSensorValueWithIndeces{\vUser}{\vTime}}}

\newcommand\vAggregate[1]{\vmath{g\left(#1\right)}}

\newcommand\vMaskingId{\vmath{\eta}}
\newcommand\vParamId{\vmath{k}}

\newcommand\vMaskingParametersIndexed{\vmath{\theta_{\vMaskingId, \vParamId}}}
\newcommand\vPrivacySetting[1]{\vmath{f_{\vMaskingId}\left(#1,
\vMaskingParametersIndexed\right)}}
\newcommand\vMaskedSet{\vmath{M}}

\newcommand\vMaskedValue{\vmath{m}}
\newcommand\vNoiseSet{\vmath{\Psi}}
\newcommand\vNoiseValue{\vmath{\psi}}

\newcommand\vLocal{\vmath{\varepsilon}}
\newcommand\vGlobal{\vmath{\epsilon}}
\newcommand\vLocalSet{\vmath{\mathcal{E}}}
\newcommand\vGlobalSet{\vmath{\mathlarger{\mathlarger{\mathlarger{\epsilon}}}}}

\newcommand\vPrivacy{\vmath{q}}

\newcommand\vPrivacySet{\vmath{Q}}

\newcommand\vUtility{\vmath{u}}

\newcommand\vUtilitySet{\vmath{U}}

\newcommand\absol[1]{\vmath{\left|#1\right|}}

\newcommand\vMean[1]{\vmath{\mu({#1})}}

\newcommand\vMedian[1]{\vmath{\text{m}\left({#1}\right)}}
\newcommand\vStandardDeviation[1]{\vmath{\sigma(#1)}}

\newcommand\vEntropy[1]{\vmath{H\left(#1\right)}}

\newcommand\vMax[1]{\vmath{\text{max}\left(#1\right)}}
\newcommand\vMin[1]{\vmath{\text{min}\left(#1\right)}}
\newcommand\vUniformVar[0]{\vmath{\upsilon}}

\newcommand\subopt[1]{\vmath{\hat{#1}}}
\newcommand\suboptrelax[1]{\vmath{#1\ssymbol{1}}}

\newcommand\reffig[1]{Figure \ref{#1}}
\newcommand\refeq[1]{(\ref{#1})}

\newcommand\refsec[1]{Section \ref{#1}}

\makeatletter
	\newcommand{\removelatexerror}{\let\@latex@error\@gobble}
\makeatother

\usepackage{verbatim}

\def\BibTeX{{\rm B\kern-.05em{\sc i\kern-.025em b}\kern-.08em
T\kern-.1667em\lower.7ex\hbox{E}\kern-.125emX}}

\def\@fnsymbol#1{\ensuremath{\ifcase#1\or *\or \dagger\or \ddagger\or
   \mathsection\or \mathparagraph\or \|\or **\or \dagger\dagger
   \or \ddagger\ddagger \else\@ctrerr\fi}}
\newcommand{\ssymbol}[1]{^{\@fnsymbol{#1}}}

\begin{document}

\begin{frontmatter}

\title{Optimization of Privacy-Utility Trade-offs under Informational Self-determination}

\author[ethz]{Thomas Asikis}
\author[ethz]{Evangelos Pournaras}

\address[ethz]{Professorship of Computational Social Science\\
ETH Zurich, Zurich, Switzerland\\
\{asikist,epournaras\}@ethz.ch
}

\begin{abstract}


The pervasiveness of Internet of Things results in vast volumes of personal data
generated by smart devices of users (data producers) such as smart phones,
wearables and other embedded sensors. It is a common requirement, especially for Big Data
analytics systems, to transfer these large in scale and distributed data to
centralized computational systems for analysis. Nevertheless, third parties that
run and manage these systems (data consumers) do not always guarantee users'
privacy.
Their primary interest is to improve utility that is usually a metric related to
the performance, costs and the quality of service. There are several techniques
that mask user-generated data to ensure privacy, e.g. differential privacy.
Setting up a process for masking data, referred to in this paper as a `privacy
setting', decreases on the one hand the utility of data analytics, while, on the
other hand, increases privacy. This paper studies parameterizations of
privacy settings that regulate the trade-off between maximum utility, minimum
privacy and minimum utility, maximum privacy, where utility refers to the
accuracy in the estimations of aggregation functions. Privacy settings can be
universally applied as system-wide parameterizations and policies (homogeneous
data sharing). Nonetheless they can also be applied autonomously by each user or
decided under the influence of (monetary) incentives (heterogeneous data
sharing). This latter diversity in data sharing by informational
self-determination plays a key role on the privacy-utility trajectories as shown
in this paper both theoretically and empirically. A generic and novel
computational framework is introduced for measuring privacy-utility trade-offs
and their Pareto optimization. The framework computes a broad spectrum of such
trade-offs that form privacy-utility trajectories under homogeneous and
heterogeneous data sharing. The practical use of the framework is experimentally
evaluated using real-world data from a Smart Grid pilot project in which energy
consumers protect their privacy by regulating the quality of the shared power
demand data, while utility companies make accurate estimations of the aggregate
load in the network to manage the power grid. Over \vmath{20,000} differential
privacy settings are applied to shape the computational trajectories that in turn provide a vast potential for
data consumers and producers to participate in viable participatory data sharing
systems.


%
\end{abstract}



\end{frontmatter}

\section{Introduction}

High data volumes are generated in real-time from users' smart devices such as
smartphones, wearables and embedded sensors. Big Data systems process these
data, generate information and enable services that support critical sectors of
economy, e.g. health, energy, transportation etc. Such systems often rely on
centralized servers or cloud computing systems. They are managed by corporate
third parties referred to in this paper as \emph{\vDc{}s} who collect the data of users
referred to respectively as \emph{\vDp{}s}. Data consumers perform data analytics for
decision-making and automation of business processes. However, data producers
are not always aware of how their data are used and processed. Terms of Use are
shown to be limited and ineffective~\cite{Bhatia2015, McDonald2008}. Security
and privacy of users' data depend entirely on \vDc{}s and as a result misuse of personal information is
possible, for instance, discrimination or limited freedom and autonomy by
personalized persuasive systems~\cite{Helbing2015, Carmichael2016,
kitchin2014data, Lekakos2014}. Giving control back to \vDp{}s by self-regulating
the amount/quality of shared data can limit these threats~\cite{Pournaras2016}.
Incentivizing the sharing of higher amount/quality of data results in improved
quality of service, i.e. higher accuracy in predictions~\cite{Soria2017,
Duan2009, Krutz2010}. At the same time, data sharing empowers \vDp{}s with an economic
value to claim.

Several applications do not require storage of the individual data generated by
\vDp{}s. Instead, \vDc{}s may only require aggregated data. For instance, Smart
Grid utility companies compute the total daily power load or the average voltage
stability to prevent possible network failures, bottlenecks, predict future
power demand, optimize power production and design pricing policies
\cite{Kursawe2011,Burns2014}. Privacy-preserving masking
mechanisms~\cite{Aggarwal2008}, i.e. differential privacy, accurately
approximate the actual aggregate values without transmitting the
privacy-sensitive individual data of \vDp{}s. Masking is a numerical
transformation of the sensor values that usually relies on the generation of
random noise and is irreversible\footnote{It is computationally infeasible to
compute the original data using the transformed data.}.

Privacy-preserving masking mechanisms are studied by calculating metrics of
privacy \vPrivacy{} and utility \vUtility{}. The former represents the amount of
personal information that a \vDp{} preserves when sharing a masked data value.
The latter represents the benefit that a \vDc{} preserves when using certain
masked data for aggregation, e.g. accuracy in data analytics. Literature
work~\cite{Aggarwal2008,Li2009,Krause2008} shows that privacy and utility are
negatively correlated, meaning that an increase on one results in decrease on the
other. This paper studies the optimization of computational trade-offs between
privacy and utility that can be used to model information sharing as
supply-demand systems run by computational markets~\cite{Pournaras2016,Li2014}.
These trade-offs can be measured by the opportunity cost between
privacy-preservation and the performance of algorithms operating on masked data,
i.e. prediction accuracy. Trade-offs can be made by choosing different
parameters for different masking mechanisms each influencing the mean or the
variance of the generated noise distributions~\cite{Aggarwal2008}. Each
parameterization results in a pair of privacy and utility values within a
trajectory of possible privacy-utility values.

The selection of parameters for masking mechanisms that maximize privacy and
utility is studied in this paper as an optimization
problem~\cite{Li2009,Krause2008}. In contrast to related work that exclusively
focuses on universal optimal privacy settings (homogeneous data sharing), this
paper studies the optimization of privacy-utility trade-offs under diversity in
data sharing (heterogeneous data sharing). This is a challenging but more
realistic scenario for participatory data sharing systems that allow
informational self-determination via a freedom and autonomy in the
amount/quality of data shared by each \vDp{}. A novel computational framework is
introduced to compute the privacy settings that realize different
privacy-utility trade-offs.

The main contributions of this article are the following: (i) The introduction
of a generalized, domain-independent, data-driven optimization framework, which
selects privacy settings that maximize privacy and utility. \vcomment{The optimization can
be constrained according to the privacy-utility requirements of \vDc{}s and
\vDp{}s. }(ii) A formal proof on how high utility can be achieved under
informational self-determination (heterogeneous data sharing) originated from
the diversity in the privacy settings selected by the users. (iii) The
introduction of new privacy and utility metrics based on statistical
properties of the generated noise. (iv) The introduction of a new masking
mechanism. (v) An empirical analysis of privacy-utility trajectories of more
than \vmath{20,000} privacy settings computed using real-world data from a Smart
Grid pilot project.

This paper is outlined as follows: Section \ref{sec:related} includes related
work on privacy masking mechanisms, privacy-utility trade-off as well as
privacy-utility maximization problems. Section \ref{sec:problem} defines the
optimization problem and illustrates the research challenge that this paper
tackles. Section \ref{sec:framework} introduces the proposed optimization
framework.
Section \ref{sec:setting} outlines the experimental settings on which the
proposed framework is tested and evaluated. Section \ref{sec:eval} shows the
results of the experimental evaluation. Finally, \refsec{sec:conclusion}
concludes this paper and outlines future work.



\section{Related Work}\label{sec:related}

Several algorithms are proposed to perform data aggregation without transmitting
the raw data. The basic idea behind such algorithms is
to irreversibly transform\footnote{A process also known as masking.} the data,
so that the original values cannot be estimated. While doing so, some of the
properties of the data should be preserved to accurately estimate aggregation
functions such as sum, count or multiplication \cite{Aggarwal2008, Duan2009,
Pournaras2016, Gentry2010, Dwork2006}. The masking process enables the \vDp{}s
to control the amount of personal information sent to \vDc{}s. These methods
also ensure that the data remain private even when a non-authorized party
acquires them, for example in the case of a man-in-the-middle attack.

\subsection{Privacy-preserving mechanisms}\label{sec:privacy-pres}

An overview of privacy-preserving mechanisms is illustrated below:

\subsubsection{Pertrubative masking mechanisms}\label{sec:per} 

Perturbative masking mechanisms allow the \vDp{}s to share their data after
masking individual values. Each value is perturbed by replacing it with a new
value that is usually generated via a process of random noise generation or
vector quantization techniques on current and past data
values~\cite{Aggarwal2008}.\vcomment{ Perturbative masking methods are often used in
privacy preserving schemes related to numerical data.} Some of the most
well-known perturbative masking methods are the following:

\textbf{Additive noise}:
A privacy-preserving approach is the addition of randomized
noise~\cite{Dwork2006,Dwork2008,Wang2013}. This approach is often used in
differential privacy schemes~\cite{Dwork2008}. Differential privacy is ensured
when the masking process prohibits the estimation of the real data values, even
if the \vDc{} can utilize previously known data values or the identity of the
individual who sends the data~\cite{Dwork2014}. Algorithms that achieve
differential privacy rely on the notion that the change of a
single element in a database does not affect the probability distribution of
the elements in the database~\cite{Dwork2006,Dwork2014,Wang2013,Wang2016}.
Furthermore, the removed element cannot be identified when comparing the version
of the database before and after the removal. This is achieved by adding a
randomly generated noise to each data value. The distribution of the random
noise is parameterized and usually is symmetric around 0
and relies on the cancellation of noises with opposite values. Increasing
the number of noise values also increases the noise cancellation, since a larger
number of opposite values are sampled. This property can be used to combine
differential privacy mechanisms in order to ensure privacy while achieving high
utility~\cite{Kairouz2017}. Statistical aggregation queries on the masked data
return an approximate numerical result, which is close to the actual result.
Differential privacy can be applied to discrete and continuous variables for the
calculation of several aggregation functions~\cite{Duan2009}.
Differential privacy can be combined with the usage of deep neural
networks~\cite{Shokri2015,Phan2016}, to apply more complex aggregation
operations on statistical databases.\vcomment{ Additive noise does not introduce extra
communication costs, since the masked value is transmitted instead of the
original value.
The storage and computational costs for implementing additive noise are often
minimal.}  \visibleComment{Furthemore, several additive noise implementations
are susceptible to noise filtering attacks, such as the use of Kalman filters
\cite{Gibson1991} or reconstruction attacks \cite{Dwork2017}. These attacks can
be prevented when the noise is not autocorrelated or the distribution of its
autocorrelation is approximately uniform.}

\textbf{Microaggregation}:
Microaggregation relies on the replacement of each data value with a
representative data value that is derived from the statistical properties of the
dataset it belongs to. A well-known application of microaggregation is
K-anonimity. K-anonymity relies on the notion that at least K original data
values are mapped to the same value~\cite{samarati-protect}. When a crisp
clustering algorithm is applied on the data, each data value is mapped to the
cluster centroid it belongs to. K is the minimum number of elements in a
cluster. Using crisp clustering techniques\footnote{Such as K-Means.} may result
in vulnerabilities to specific attacks, so membership or fuzzy clustering is
preferred instead~\cite{Nin2008}.
Membership clustering assigns a data point to multiple clusters with a
probability that is often proportional to the distance from each cluster
centroid. For membership clustering techniques, usually large amounts of data
are required. The storage and computational capacity of sensor devices cannot
usually support such processes \cite{Nin2008, Aggarwal2008}.

\textbf{Synthetic microdata generation}
An new dataset is synthesized based on the original data and multiple
imputations~\cite{Aggarwal2008}. The ``synthetic'' dataset is used instead of
the original one for aggregation calculations.\vcomment{ This paper focuses on Internet of
Things and the use of sensors. As a result the whole dataset is not available
for processing at runtime as storing older sensor values and timestamps on the
sensor device often results in high storage costs.} The application of synthetic
microdata generation on sensor devices may produce prohibitive processing and
storage costs. Furthermore, the availability of historical data on each sensor
device may not be adequate for such methods to achieve comparable performance
and efficiency with the perturbative masking methods~\cite{Aggarwal2008}.

\subsubsection{Encryption}
Several approaches use encryption to produce an encrypted set of numbers or
symbols, known as ciphers. The aggregation operations can be performed on the
ciphers and produce an encrypted aggregation value. The encrypted aggregate
value can then be decrypted to the original aggregate one, with the usage of the
corresponding private and public keys and decryption schemes, providing maximum
utility and privacy to the recipient. The encrypted individual values cannot be
transformed to the original values without the usage of the appropriate keys
from an adversary, so maximum privacy is ensured. Currently, there is extensive
research on this area, and there has been a recent breakthrough with the
development of fully homomorphic encryption
schemes~\cite{gentryThesis,Gentry:2009:FHE:1536414.1536440,Gentry2010,Gentry2011}.
Homomorphic encryption schemes though require high computational and
communication costs, especially when applied in large scale
networks~\cite{Gentry2009,duan2009differential}.

\subsubsection{Multi-party computation}
Multi-Party Computation (MPC)~\cite{Yao:1982:PSC:1398511.1382751, BennatiP17} can also be
used for privacy-preservation~\cite{Du2001} by moving data from one device to
another. In such an approach, security and integrity of the data depend on the
resilience and security of the network.\vcomment{ Furthermore, MPC requires the design and
application of certain protocols.} Most of the methods that rely on encryption
can calculate the exact sum of the data, but they can also be violated if an
attacker manages to have access to the private key or uses an algorithm that can
guess it. Furthermore, in most cases they rely on communication protocols\vcomment{ and complex
computational schemes} that burden the system with extra computational and
communication costs \cite{Prabhakaran2008}. These costs are often
prohibitive for devices such as IoT sensors and smartphone wearables in which computational
power and storage are limited~\cite{BennatiP17}.

\subsection{Privacy and computational markets}
A supply-demand system operating on a computational market of data, can be
created with the introduction of self-regulatory privacy-preserving information
systems~\cite{Pournaras2016}. Privacy preservation is utilized to create such
systems, for instance by using K-means for microaggregation and different numbers of
clusters for each sensor. Varying the number of clusters produces different
levels of privacy and utility.\vcomment{ Privacy is measured by evaluating the difference
between masked and original data values. Utility is measured by evaluating the
difference between the aggregate of the original and masked data values.} The
resulting trade-off between privacy and utility is used to create a reward
system, where \vDc{}s offer rewards for the data provided by the \vDp{}s. The
rewards are based on the demand of transformed data that enables the estimation
of more accurate aggregate values.

A reward system can be combined with pricing strategies from existing literature
on pricing private data~\cite{Li2014}, in which three actors are introduced:\vcomment{ (i)
the \textit{Data Buyer}, who represents a \vDc{} (ii) the \textit{Data Owner},
who corresponds to a \vDp{} and (iii) the \textit{Market Maker}, who creates the
market in which the first two operate.} Various pricing functions are proposed to
the \textit{Market Maker} so that the privacy-utility of both \vDc{}s and
\vDp{}s are satisfied. The optimization framework of the current paper can
utilize any parametric masking mechanism of the literature mentioned in
\refsec{sec:privacy-pres}. The output of the optimization can be used along with
pricing functions on participatory computational markets, to create fully
functional and self-regulatory data markets.

\subsection{Comparison and positioning}
The challenge of an automated selection of privacy settings that satisfy
different trade-offs is not tackled in the aforementioned mechanisms.
Privacy-utility trajectories have not been earlier studied extensively and
empirically as in the rest of this paper. The optimization of privacy-utility
trade-offs under diversity in data sharing originated from informational
self-determination is the challenge tackled in this paper. To the best of the
authors' knowledge, this challenge is not the focus of earlier work.

 
\section{Problem Definition}\label{sec:problem}
Related work \cite{Krause2008, Soria2017, Li2009, Sanchez2014, Pournaras2016} on
privacy-utility trade-offs focuses on the parameter optimization of a single
masking mechanism. A masking mechanism is often a noise generation process,
which samples random noise values from a laplace distribution and then it
aggregates it to the data, for instance the sampled noise is then added to the
data to achieve differential privacy~\cite{Dwork2006}.
The result of the optimization is usually a vector of parameter values
\vMaskingParametersIndexed, for a masking mechanism \vMaskingId and parameter
index \vParamId. The pair of the masking mechanism and the parameter values is
referred as a \textit{privacy setting} \vPrivacySetting{\vSensorValueSet} of a
set of sensor values \vmath{\vSensorValueSet{}\,\in\,R^{1}}. This privacy
setting produces a pair of privacy-utility values
\vmath{\subopt{\vPrivacy}\,,\subopt{\vUtility}}, such that:
\begin{equation}
\label{eq:maxPrivacy}
\begin{gathered}
	\subopt{\vPrivacy}\,\to\,\vMax{\vPrivacySet}
\end{gathered}
\end{equation}
\begin{equation}
\label{eq:maxUtility}
\begin{gathered}
	\subopt{\vUtility}\,\to\,\vMax{\vUtilitySet}
\end{gathered}
\end{equation}
Where \vmath{(\subopt{\vPrivacy}, \subopt{\vUtility})} is a (sub-optimal)
privacy-utility pair of values, which is computed by an optimization algorithm
that searches for the optimal privacy-utility values pair.
\vmath{\vMax{\vPrivacySet}}, \vmath{\vMax{\vUtilitySet}} are the maximum privacy
and utility values of a privacy value set \vPrivacySet{} and a utility value set
\vUtilitySet{}. These sets are generated by the application of a masking
mechanism.

The optimization of an objective function that satisfies both Relations
(\ref{eq:maxPrivacy}) and (\ref{eq:maxUtility}) simultaneously is an NP-hard
problem \cite{Krause2008}, \visibleComment{in the case that privacy and utility
are orthogonal (\vmath{\vPrivacy \perp \vUtility}) or
opposite\footnote{\visibleComment{In the case that privacy and utility are
positive correlated (\vmath{\vPrivacy \uparrow, \vUtility \uparrow}), the
problem is reduced to NTIME-hard, and especially in the case privacy and utility
are proportional \vmath{\vPrivacy \propto \vUtility} to DTIME-hard
\cite{BOOK1974213}. The solution of the problem is provided by linearly
evaluating all pairs of privacy and utility values once without comparing to all
other pairs.}} (\vmath{\vPrivacy \uparrow, \vUtility \downarrow})}, and often
intractable to solve, since
privacy-utility trade-offs prohibit the satisfaction of both Relations
(\ref{eq:maxPrivacy}) and (\ref{eq:maxUtility}).
Particularly, maximizing simultaneously utility and privacy usually yields
sub-optimal values, which are lower than the corresponding optimal values
computed by optimizing each metric separately \cite{Krause2008}.
Furthermore, such optimization is applicable for statistical databases
\cite{Dwork2014, Aggarwal2008}, where data are stored in a centralized system.
In such case, a specific privacy setting is chosen by the designer/administrator
of the system. As a result, this approach relies on the assumption that a
specific privacy setting should be used by all \vDp{}s.

However, remaining to a fixed privacy setting may be limited for \vDp{}s,
especially when a \vDp{} wishes to switch to a different privacy setting to
improve privacy further. In this case, the optimization of different objective
functions is formalized in the following inequalities:

\begin{equation}
\label{ineq:privacyIneq}
\begin{gathered}
	\suboptrelax{\vPrivacy}\,>\,\subopt{\vPrivacy}  +
	\delta 
	\land
	\suboptrelax{\vUtility}\,>\,\subopt{\vUtility} + c
\end{gathered}
\end{equation}
Where \vmath{\delta} measures the change in privacy, which denotes whether the
\vDp{}s require higher privacy, \vmath{\delta > 0}, or lower privacy
\vmath{\delta < 0}, from the system. \vmath{c} measures the change in utility,
which denotes whether the \vDc{} demands lower utility, \vmath{c > 0}, or higher
utility \vmath{c < 0}, from the system. Finally,
\vmath{(\suboptrelax{\vPrivacy}, \suboptrelax{\vUtility})} denotes a new
(sub-optimal) pair of privacy-utility values, computed by an optimization
algorithm that searches for the optimal pair of privacy-utility values with
respect to the privacy requirements of \vDp{} and the utility requirements of
\vDc{} expressed by \vmath{c} and \vmath{\delta} respectively.

The optimization of an objective function to satisfy Relation
(\ref{ineq:privacyIneq}) is also based on the assumption that all \vDp{}s
agree to use the same privacy setting. This means that \vDp{}s may acquire a different privacy level by changing the value of \vmath{\delta} via the collective selection of a different privacy setting. Consequently, a single privacy
setting is generated and it produces a  pair of privacy-utility values, which
satisfy Inequality (\ref{ineq:privacyIneq}). The value of \vmath{\delta} is
determined via a collective decision-making process applied by the \vDp{}s, e.g.
voting between different privacy-utility requirements.
Such a system is referred to as a \textit{homogeneous} privacy system, where
\vDp{}s are able to influence the amount of privacy applied on the data by
actively participating in the market, nevertheless they all share the same value for
\vmath{\delta}.
The \vDc{} can bargain for higher utility by offering higher rewards to the
\vDp{}s to lower their privacy requirements.

Another challenge that arises is the optimization between privacy and
utility when each user decides and self-determines a preferred
privacy setting instead of using a universal privacy setting. In such a
scenario, inequality (\ref{ineq:privacyIneq}) is substituted by the following set of inequalities:

\begin{equation}
\label{ineq:heterogeneous}
\begin{gathered}
	(\suboptrelax{\vPrivacy}_{1}>\subopt{\vPrivacy}  +
	\delta_{1})\land
	\ldots\land\\
	(\suboptrelax{\vPrivacy}_{\vUser{}}>\subopt{\vPrivacy}  +
	\delta_{\absol{\vUserSet}})\land(
	\suboptrelax{\vUtility}>\subopt{\vUtility} + c)
\end{gathered}
\end{equation}
Where \vmath{\delta_{n}} measures the change in privacy which denotes whether a
\vDp{} \vmath{n} belonging to a set of users \vmath{\vUserSet} requires higher
privacy, \vmath{\delta_{n} > 0}, or lower privacy \vmath{\delta_{n} < 0}.
\vmath{\suboptrelax{\vPrivacy}_n} denotes a new (sub-optimal) privacy value for
each \vDp{} \vUser{}. The value is computed by an optimization algorithm that
searches for the optimal privacy value with respect the \vDp{}'s privacy
requirements expressed by \vmath{\delta_{\vUser}}.

A system in which the inequalities of Relation (\ref{ineq:heterogeneous}) hold
is referred to as an \textit{heterogeneous} privacy system, where each \vDp{}
self-determines and autonomously applies a privacy setting based on a preferred
privacy value and an expected reward for increasing system utility.\vcomment{ Concluding,
based on the above, the challenges that this article addresses are the
following: (i) The design of a computational framework for the optimization of
privacy-utility trade-offs in homogeneous as well as heterogeneous information
sharing systems. (ii) The analytical and empirical study of these trade-offs.}
 
\section{Framework}\label{sec:framework}
The design of a new privacy preserving optimization framework is introduced in
this section to tackle the challenges posed in \refsec{sec:problem}. Additive
noise masking mechanisms require a lower number of parameters in general and they are often
used in privacy-utility optimization~\cite{Aggarwal2008, Dwork2014, Krause2008}.
Each privacy setting is illustrated as an ellipse\footnote{The elliptical shape
is chosen for the sake of illustration and it indicates a symmetrical
distribution of privacy-utility values, generated by a privacy setting, within the ellipse area.} in
\reffig{fig:optim:a}.
Each point within the ellipse is a possible privacy-utility pair of values. The
ellipse center is chosen based on the privacy and utility mode of the setting.
The mode is the value with the highest density. In symmetric distributions, it
can be measured via the mean.
The vertical radius of the ellipse denotes the dispersion of utility values,
while horizontal radius denotes the dispersion of privacy values.
Additive noise is stochastic, which means that applying the same privacy setting
on the same dataset yields varying privacy-utility values. The choice of an
optimal privacy-utility pair cannot be achieved by only evaluating the mode of privacy and utility for each privacy setting. If the
privacy-utility values of a privacy setting with high utility mode are varying
to a large extend, there is high probability that unexpected non-optimal
values are observed.
To overcome this challenge, the objective function of the parameter optimization
algorithm selects the parameters that minimize the dispersion\footnote{This
refers to the dispersion measures of the privacy and
utility distributions. If the values belong to a gaussian distribution, then the standard deviation is
used to measure the dispersion. Since this is not always the case, other
measures of scale can be used, such as the Inter-Quantile Range(IQR).} of
privacy-utility values while maximizing the expected utility.

\begin{figure}[ht!]
  \subfloat[Privacy-utility trajectory]{%
       \includegraphics[width=0.45\linewidth]{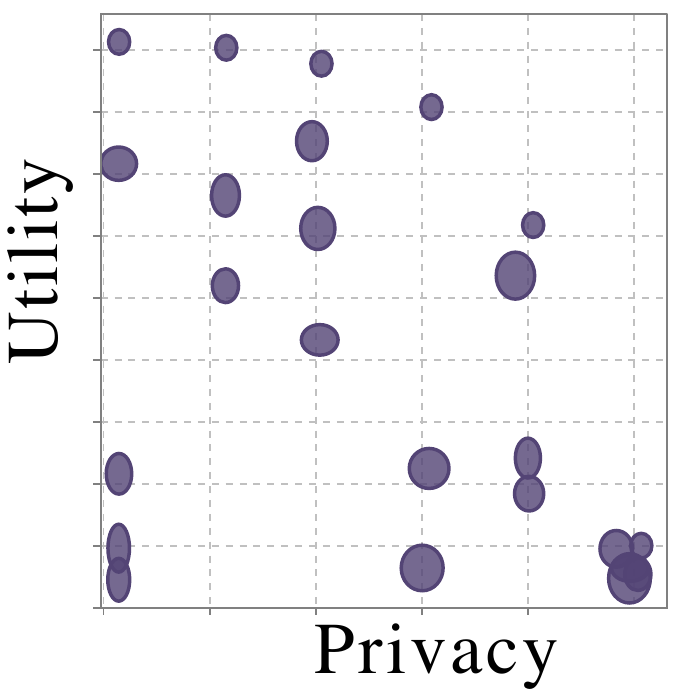}
       \label{fig:optim:a}}\hfill
  \subfloat[Binning of the privacy range]{%
        \includegraphics[width=0.45\linewidth]{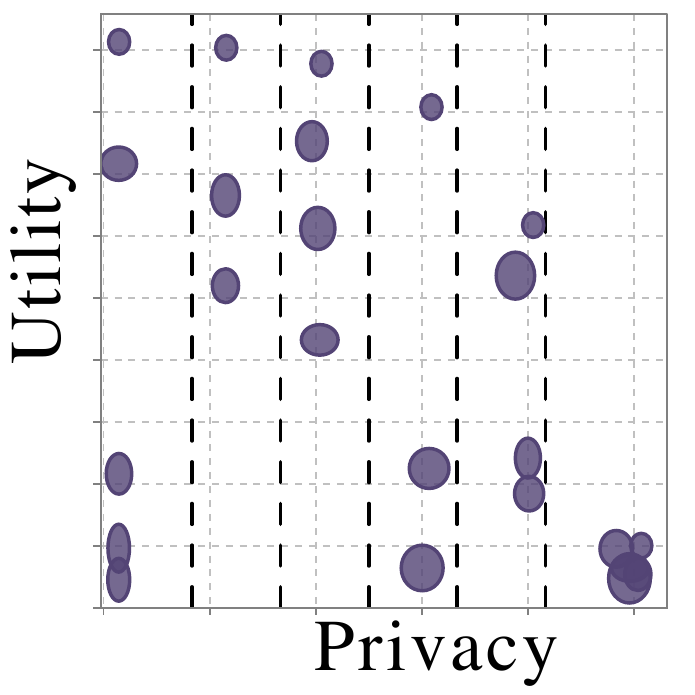}
        \label{fig:optim:b}}\hfill\\
  \subfloat[Evaluation via objective function]{%
        \includegraphics[width=0.45\linewidth]{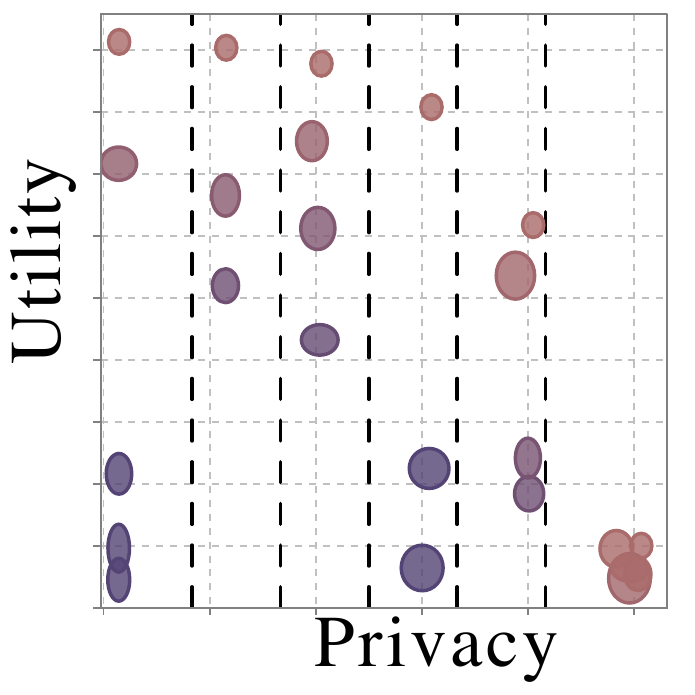}
        \label{fig:optim:c}}\hfill
  \subfloat[Bin optimization]{%
        \includegraphics[width=0.45\linewidth]{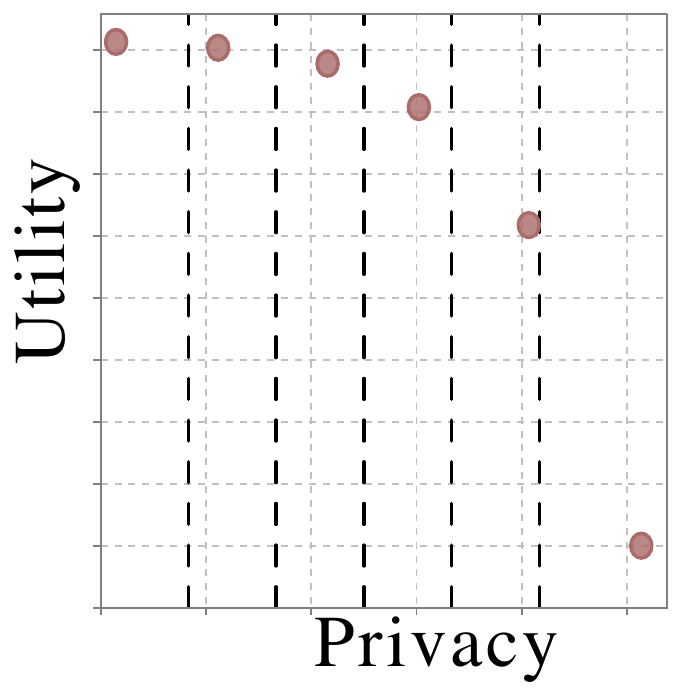}
        \label{fig:optim:d}}\hfil\\
  \subfloat[Objective function scale]{%
        \includegraphics[width=1.0\linewidth]{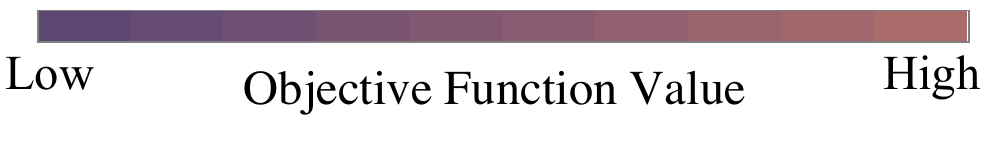}
        \label{fig:optim:e}}\hfil\\
  \caption{A graphical representation of the algorithm. Each ellipse denotes the
  privacy-utility values of a privacy setting. In Figures
  \ref{fig:optim:c} and \ref{fig:optim:d} the varying color denotes the fitness
  value. A lighter red color denotes higher fitness.}
  \vspace{-0.27in}
  \label{fig:stats:global} 
\end{figure}

A \vDp{} selects any privacy setting, among different ones, that satisfies
personal privacy requirements. The proposed framework divides the range of
privacy values in a number of equally sized bins, as illustrated in
\reffig{fig:optim:b}. Within each bin, a fitness value is calculated for each
privacy setting, based on privacy-utility mode and dispersion Each privacy
setting produces privacy values with low dispersion. This is done by applying a
lower bound constraint on privacy and utility constraint on the dispersion of
privacy values and evaluating only privacy settings that satisfy this
constraint, as shown in \reffig{fig:optim:c}.
The optimization framework evaluates several privacy settings, to find the
parameters that achieve maximum privacy-utility values that vary as little as
possible. This is illustrated in \reffig{fig:optim:d} in which the ellipses with
the highest utility mode and lowest utility dispersion are filtered for each
privacy bin.
\vcomment{The resulting privacy settings are then provided to the \vDp{}s. }

In a homogeneous data sharing system, a universal privacy setting is selected by
the \vDp{}s, via, for instance, voting~\cite{nurmi2012comparing}. Alternatively,
in a heterogeneous system, the \vDp{}s self-determine the privacy setting
independently. Theorem~\ref{thm:heterogeneous} below proves that aggregation
functions can be accurately approximated (utility can be maximized) even if
different privacy settings from the same of different masking mechanisms are
selected.

\begin{thm}\label{thm:heterogeneous}

Let the transformation of \absol{I} disjoint subsets of sensor values
\vmath{\vSensorValueSet_i} into the respective subsets of masked values
\vmath{\vMaskedSet_i} using a certain privacy settings \vmath{f_i}
for each such transformation. It holds that the aggregation of the generated
multisets of masked values \vmath{\vMaskedSet_i} approximates the aggregation of
the sensor values multiset \vmath{\vSensorValueSet_i}:

\begin{equation}\label{eq:masking}
g(\bigcup_{i=1}^{\absol{I}}\vMaskedSet_i) \to g(\vSensorValueSet),
\end{equation}

\noindent given that the commutative and associative properties hold between
each of the privacy settings \vmath{f_i} and the aggregation function \vmath{g}.

\end{thm}


\begin{proof}
Let a multiset of real sensor values $\vSensorValueSet \subseteq R^{1}$ and
$\absol{I}$ disjoint subsets of $ \vSensorValueSet $ such that:
\begin{equation} 
\label{eq:subsets}
	\bigcup_{i=1}^{\absol{I}} \vSensorValueSet_i = \vSensorValueSet,
	\vSensorValueSet_i \neq \emptyset \,\, \forall{}\, i
	\in \{1,...,\absol{I}\}
\end{equation}
Let a privacy setting $ f:\vSensorValueSet{}, \vNoiseSet{} \to \vMaskedSet $ be
a pairwise element operation between a set of sensor values \vSensorValueSet and a
set of noise values \vNoiseSet, that transforms each sensor value $\vSensorValue
\in \vSensorValueSet{}$ \visibleComment{by aggregating it} with a randomly
selected noise value \vNoiseValue from \vNoiseSet to produce a masked value \vMaskedValue:
\begin{equation}\label{eq:masking}
	\begin{gathered}
		f(\vSensorValueSet, \vNoiseSet) \visibleComment{= g(\vSensorValueSet \cup
		\vNoiseSet)} = \vMaskedSet \Leftrightarrow\\
		f(\vSensorValue, \vNoiseValue) \visibleComment{ = g(\{\vSensorValue,
		\vNoiseValue\})} = \vMaskedValue
	\end{gathered}
\end{equation}

Let  $ g:A \to R^{1} $ be an aggregation function which aggregates all
elements of real values multisets $
\vSensorValueSet,\,\vNoiseSet,\,\vMaskedSet \subseteq A \subseteq R^{1} $ into a
single real value $g(A) = z_A \in R^{1}$\vcomment{ and can also be expressed as a
pairwise recursive operation between elements of those sets.
The aggregation function is defined as $g(A) = z_A$}.
Assume that \vmath{g:A\,\to\,R^{1}} is defined in a recursive manner so that it
satisfies the following equation for a multiset $A$ and any union of all
possible combinations of disjoint subsets $A_i$ that satisfy Relation
\refeq{eq:subsets}:
\begin{equation} 
\label{eq:recursive}
	\begin{gathered}
		g(A) = g(\bigcup_{i=1}^{\absol{I}}A_i) = g( \bigcup_{i=1}^{\absol{I}}g(A_i))\\
	\end{gathered}
\end{equation}
\visibleComment{According to literature \cite{Beliakov2007AggregationFA} the
family of aggregation functions that Relation \ref{eq:recursive} applies to
is referred to as extended aggregation functions\footnote{\visibleComment{A
subset of those functions are the averaging functions, which include aggregations such as the mean, weighted
mean, Gini mean, Bonferoni mean, Choquet integrals etc.}}.} The pairwise
operation between \vSensorValue and \vNoiseValue in \vmath{f} is designed in
such way that it satisfies the commutative and associative properties when
combined with the pairwise operation of \vmath{g}:
\begin{equation}\label{eq:distributive}
	\begin{gathered}
		g(f(\vSensorValueSet, \vNoiseSet)) \stackrel{\ref{eq:masking}, \ref{eq:recursive}}{=} f(g(\vSensorValueSet),
		g(\vNoiseSet))\\
	\end{gathered}
\end{equation}
where \vmath{g(\vNoiseSet) \to \iota}, \vmath{\iota} is the strong neutral element
of the extended aggregation function \vmath{g}, such that:
\begin{equation}\label{eq:identity}
	\begin{gathered}
 		\vmath{g(g(A)\cup{}\iota)\,=\,g(A)}\,\Rightarrow\\
 		g(g(A)\cup{}g(\vNoiseSet))\,\to\,g(A)
 	\end{gathered}
 \end{equation}
This property is used in the noise cancellation of \refsec{sec:per}. Let
$\absol{I}$ multisets $\vNoiseSet_i$ of noise that satisfy Relation \ref{eq:subsets}, then the following relation holds:
\begin{equation}
\label{eq:break}
	\begin{gathered}
		g(\vMaskedSet_{i}) = g(f(\vSensorValueSet_i,\vNoiseSet_i))
		\,\stackrel{(\ref{eq:distributive})}{\Leftrightarrow}\\
		g(\vMaskedSet_{i}) = f(g(\vSensorValueSet_i),g(\vNoiseSet_i))
		\,\stackrel{(\ref{eq:identity}, \ref{eq:masking})}{\Leftrightarrow}\\
		g(\vMaskedSet_{i})
		\to g(\vSensorValueSet_i),
	\end{gathered}
\end{equation}
which means that each noise multiset $\vNoiseSet_i$ is generated in such a way that
the aggregation of $g(\vMaskedSet_i)$ approximates the
aggregation of $g(\vSensorValueSet_i)$. An illustrative example is
the laplace noise used in the literature for the aggregation
functions of count or summation \cite{Dwork2006,duan2009differential}, which
satisfies Relations \ref{eq:masking}, \ref{eq:recursive} and \ref{eq:identity}.
Now it can be proven that:
\begin{equation}\label{eq:proof}
	\begin{gathered}
		g(\bigcup_{i=1}^{\absol{I}}f_i(\vSensorValueSet_i,\vNoiseSet_i))
		\stackrel{\refeq{eq:recursive}}{=}
		g(\bigcup_{i=1}^{\absol{I}}g(f_i(\vSensorValueSet_i,\vNoiseSet_i)))\stackrel{\refeq{eq:break}}{\iff}\\
		g(\bigcup_{i=1}^{\absol{I}}\vMaskedSet_i) \to
		g(\bigcup_{i=1}^{\absol{I}}g(\vSensorValueSet_i)))\stackrel{\refeq{eq:subsets},
		\refeq{eq:recursive}}{\iff}\\
		g(\bigcup_{i=1}^{\absol{I}}\vMaskedSet_i) \to g(\vSensorValueSet)
	\end{gathered}
\end{equation}
Thus, Theorem \ref{thm:heterogeneous} is proven.
\end{proof}


The practical implication of Theorem \ref{thm:heterogeneous} is that the
aggregation of sensor values is approximated by the aggregation of masked values
produced by different privacy settings. The approximation stands as long as the
noise values produced by the different privacy settings satisfy Relations
\ref{eq:distributive} and \ref{eq:identity}. According to Relation
\ref{eq:subsets}, each subset of sensor values should be masked by one privacy
setting. \visibleComment{Regarding the complexity of these operations, applying
the masking on top of sensor values is linearly depended to the number of sensor
values \vmath{\absol{\vSensorValueSet{}_{i}}} assigned to each privacy setting.
Due to Relation \ref{eq:subsets}, applying the proposed framework in real time
increases computational complexity by \vmath{O(\absol{\vSensorValueSet})}.
The original values are not stored or transmitted at runtime, thus the storage
and communication complexity does not change. During optimization all the
privacy settings \vmath{i \in I} are applied to a training set of sensor values
\vmath{S}. In that case real sensor values are stored and transmitted as well
along with the masked values for each setting. The storage and communication
costs increase by \vmath{O(\absol{I}\cdot{}\absol{\vSensorValueSet{}})}. The
computation costs also increase to
\vmath{O(\absol{I}\cdot{}\absol{\vSensorValueSet{}})}, which is a quadratic
complexity in the worst case \vmath{\absol{I}=\absol{\vSensorValueSet{}}}. In
most real world applications, it is safe to assume that the sensor values have
considerably higher volume to the evaluated privacy settings
\vmath{\absol{I}<<\absol{\vSensorValueSet{}}}, thus the expected computational,
storage and communication complexity are linear to the number of sensor values.}

\visibleComment{The framework can be applied as a multi-agent system. It
requires two types of agents representing the \vDc{}s and
\vDp{}s. This scheme can be applied in both centralized and decentralized aggregation services, 
such as MySQL or DIAS \cite{7921015}.} Finally in
both heterogeneous and homogeneous systems, the \vDc{} can influence the
\vDp{}'s choice by offering a higher amount of reward to achieve a higher
utility.

 
\section{Experimental Settings}\label{sec:setting}
This section illustrates the experimental settings, which are used to
empirically evaluate the proposed framework. A set of sensor values
\vSensorValueSet{} is used for the evaluation. Each sensor value
\vSensorValueIndexed{} belongs to a user \vUser{} and is generated at time
\vTime{}.\vcomment{ Time \vTime{} may indicate a single time point or a time period of
several time points.} For each sensor value, a privacy setting that operates on
the device of the \vDp{} masks the sensor value
\vPrivacySetting{\vSensorValueIndexed{}} by using the masking mechanism
\vMaskingId with parameters \vMaskingParametersIndexed. 
Two metrics are used to evaluate privacy and utility.

\subsection{Privacy evaluation}\label{sec:privacyEval}
The main metric, which is used to calculate privacy, is the difference of
the masked value and the original value, which is defined as the local error:

\begin{equation}\label{eq:local}
	\begin{gathered}
		\vLocal_{\vUser, \vTime} = \left|\dfrac{\vPrivacySetting{\vSensorValueIndexed{}} -
		\vSensorValueIndexed{}}{\vSensorValueIndexed{}}\right|
	\end{gathered}
\end{equation}

For a privacy setting to achieve a high privacy, a \vDc{} should not be able to
estimate the local error for the sensor values sent by \vDp{}s. This is achieved
by choosing privacy settings that generate noise that is difficult to estimate.
As it is shown in the literature \cite{Pournaras2016, Krause2008, Dwork2014,
Aggarwal2008}, the noise is difficult to estimate, if it is highly random and
causes a significant change in the original value.
\visibleComment{To avoid noise fitering attacks, noise with low or no
autocorrelation is generated. The range of autocorrelation values can be
determined analytically when the noise generation function is defined. In case
this is not possible, a metric quantifying the color of noise can be
included in the objective function.} Randomness is evaluated by measuring
the Shannon entropy \cite{Shannon1948} \vEntropy{\vLocalSet} of the local error
for all local error values \vLocalSet{}. The entropy
is calculated by creating a histogram of the error values and then applying the
discrete Shannon entropy calculation. Each bin of the histogram has a size of
0.001. The significance of change is measured by calculating the mean local
error \vMean{\vLocalSet} and standard deviation \vStandardDeviation{\vLocalSet}.
When comparing privacy settings, higher mean,
variance and entropy indicate higher privacy \cite{Aggarwal2008}. In this
article, the objective function that measures privacy for a privacy setting
\vmath{f_{\vMaskingId, \vParamId}} is defined as
follows\footnote{\label{fn:error}The error function described in
\refeq{eq:local} and \refeq{eq:global} is also known in literature as absolute
percentage error (APE)~\cite{MAKRIDAKIS1993527}. The error values are easy to
interpret, as APE measures the relative change of the sensor values and
aggregate values by using masking. Yet, when the denominator of the function is
approaching zero, then the absolute relative error cannot be calculated. If the
sensor values are sparse, then another error function can be used, such as
MAPE.}:
\begin{equation}\label{eq:privacy}
	\begin{gathered}
		\vPrivacy =
		\alpha_{1}\dfrac{\vMean{\vLocalSet_{\vMaskingId,\vParamId}}}{\vMax{\vMean{\vLocalSet_{\vMaskingId,\vParamId}}}}
		+
		\alpha_{2}\dfrac{\vStandardDeviation{\vLocalSet_{\vMaskingId,\vParamId}}}{\vMax{\vStandardDeviation{\vLocalSet_{\vMaskingId,\vParamId}}}}
		\\
		+\,
		\alpha_{3}\dfrac{\vEntropy{\vLocalSet_{\vMaskingId,
		\vParamId}}}{\vMax{\vEntropy{\vLocalSet_{\vMaskingId,\vParamId}}}}
	\end{gathered}
\end{equation}
Where \vmath{\alpha_{1}, \alpha_{2}, \alpha_{3}} are weighting parameters used
to control the effect of each metric in the privacy objective function.
\vmath{\vMax{\bullet}} is the maximum observed value for a metric during the
experiments. This value is produced by evaluating all privacy settings
\vmath{f_{\vMaskingId,\vParamId}}. Dividing by this value, normalizes the
metrics in $\left[0,1\right]$, so that the objective function is not affected by
the scale of the metric.

\subsection{Utility evaluation}\label{sec:utilEval}
The utility of the system is estimated by measuring the error the
system accumulates within a time period, by computing an aggregation function
\vAggregate{\bullet} on the masked sensor values. Examples of such
aggregation functions are the daily total, daily average and weekly variance of
the sensor values. The accumulated error is referred to as global
error\textsuperscript[\ref{fn:error}] and is defined as:

\begin{equation}\label{eq:global}
	\begin{gathered}
		\vGlobal_{\vTime} = \left|\dfrac{\vAggregate{\vMaskedSet_{\vTime}} -
		\vSensorValueSet_{\vTime}}{\vAggregate{\vSensorValueSet_{\vTime}}}\right|
	\end{gathered}
\end{equation}

A sample set of global error values \vmath{\vGlobalSet{}} is created by
applying the masking process for a number of time periods of the dataset.
\visibleComment{The mean, entropy and variance of the global error of a privacy
setting \vmath{f_{\vMaskingId,\vParamId}} is calculated over this sample.
The mean global error \vMean{\vGlobalSet_{\vMaskingId,\vParamId}} indicates the
expected error between the masked and actual aggregate. The standard deviation
\vStandardDeviation{\vGlobalSet_{\vMaskingId,\vParamId}} and the entropy
\vEntropy{\vGlobalSet_{\vMaskingId,\vParamId}} of the global error, indicate how
much and how often the masked aggregate diverges from the expected value.
Minimizing all three quantities to 0, ensures that the masked aggregate
approximates the actual aggregate efficiently.} Thus, after the global
error sample is created for each privacy setting, the corresponding utility
objective function is calculated:

\begin{equation}
\label{eq:utility}
	\begin{gathered}
		\vUtility\!=\!1\!-\!\left(
		\gamma_{1}\dfrac{\vMean{\vGlobalSet_{\vMaskingId,\vParamId}}}{\vMax{\vMean{\vGlobalSet_{\vMaskingId,\vParamId}}}}
		\!+\!
		\gamma_{2}\dfrac{\vStandardDeviation{\vGlobalSet_{\vMaskingId,\vParamId}}}{\vMax{\vStandardDeviation{\vGlobalSet_{\vMaskingId,\vParamId}}}}\right.\\\left.
		+\,
		\gamma_{3}\dfrac{\vEntropy{\vGlobalSet_{\vMaskingId,\vParamId}}}{\vEntropy{\vMax{\vGlobalSet_{\vMaskingId,\vParamId}}}}
		\vphantom{\int_1^2}
		\right)
	\end{gathered}
\end{equation}
Where the weighting parameters \vmath{\gamma_{1}, \gamma_{2}, \gamma_{3}} are
used to control the effect of each metric in the utility objective function.
\vmath{\vMax{\bullet}} is the maximum observed value for a metric during the
experiments. This value is produced by evaluating all privacy settings
\vmath{f_{\vMaskingId,\vParamId}}. Dividing by this value, normalizes the
metrics in $\left[0,1\right]$, so that the objective function is not affected by
the scale of the metric.

Recall from \refsec{sec:framework} that utility and privacy vary, when
repeating the masking process for the same privacy setting and dataset due to the randomness
of the noise. A large sample to measure this variance is created, by
applying each privacy setting over three times on the same dataset. Then
the framework of \refsec{sec:framework} filters the privacy settings
based on the mode and the scale of the privacy-utility sample, as illustrated
in \reffig{fig:optim:c}.
The privacy-utility samples for a privacy setting may not follow a symmetrical or
normal distribution\footnote{It is confirmed in some experimental settings
that some privacy settings generate samples of privacy-utility values that do
not pass a Kolmogorov Smirnoff normality test \cite{Rao1975}, and are also
non-symmetrical.}. As a result, the maximization of the following objective
function is based on utility:
\begin{equation}
\label{eq:utilmax}
	\begin{gathered}
			\text{perc}\left(\vUtilitySet,
			50\right) + \text{perc}\left(\vUtilitySet, 10\right)
	\end{gathered}
\end{equation}
Where \vmath{perc(U, i)} calculates the \vmath{i^{th}} percentile of a set of utility values \vmath{U} produced by the application of a
privacy setting.

The factors that maximize Relation (\ref{eq:utilmax}) are: (i) the value of the
mode, which is assumed to be approximated by the median and (ii) the dispersion
towards values lower than the median, which is expressed by adding the
\vmath{10^{th}} percentile to the median. The objective function evaluates the
median and the negative dispersion (\vmath{10^{th}} percentile) of utility
values.
Positive dispersion is not taken into account in the optimization, since the
abstract objective of the optimization is to ensure the least expected utility
of a privacy setting for the \vDc{}s. The privacy is constrained by evaluating
only privacy settings in which the
\vmath{10^{th}} percentile differs from the privacy median for at most
\vmath{\omega}, as shown in Inequality (\ref{eq:ptivcon}). The value of
\vmath{\omega} is constrained to be lower or equal to the bin size of the
optimization to ensure low privacy dispersion:
\begin{equation}
\label{eq:ptivcon}
	\begin{gathered}
			\text{perc}\left(\vPrivacySet,
			50\right) - \text{perc}\left(\vPrivacySet, 10\right) < \omega,
	\end{gathered}
\end{equation}
Where \vmath{perc(Q, i)} calculates the \vmath{i^{th}} percentile of a set
a set of privacy values \vmath{Q} produced by the application of a privacy
setting.

 
\section{Experimental Evaluation}\label{sec:eval}
The proposed framework is evaluated experimentally by applying it to a
real-world dataset. Privacy and utility are evaluated using over
\vmath{20,000} privacy settings for empirical evaluation.

\subsection{Electricity Customer Behavior Trial dataset} 
The Electricity Customer Behavior Trial (ECBT) dataset contains sensor data
that measure the energy consumption for \vmath{6,435} energy \vDp{}s.
The data are sampled every \vmath{30} minutes daily for \vmath{536} days. For
the proposed framework, a set of sensor values \vSensorValueSet of
\vmath{\absol{\vUserSet}=6,435} users and \vmath{|T| = 536} time periods. The
total number of sensor values in the set is
\vmath{\absol{\vSensorValueSet}=165,559,680}. The sensor data are considered
private and the utility company managing the energy network uses them to
calculate daily total consumption in the grid, to predict possible failures
and plan power production. The daily total consumption is an aggregation that
can be defined as the sum of all the sensor values generated during the day:
\vmath{\vAggregate{\vSensorValueSet_{\vTime}}=\sum_{\vUser=1}^{6435}\vSensorValueIndexed}.
Around \vmath{10\%} of the daily measurements are missing values, and are not
included in the experiments. The significance of the missing values reduces as
the aggregation interval increases. Therefore, a daily summation is chosen
over more granular summation.

During the experiments, the local error of Relation \refeq{eq:local} results in
a non-finite\footnote{The original sensor value is zero, therefore the result of
Relation \refeq{eq:local} is infinite for non-zero noise or indefinite for
zero-noise.} number only for a low number of maskings.
Hence, these values are excluded from the experiments, so that the calculation
of finite local error values is feasible. Concluding, the proposed framework
operates on \vmath{90\%} of the ECBT dataset.

\subsection{Privacy mechanisms}\label{sec:mechs}
Among several masking mechanisms~\cite{Aggarwal2008}, two ones are used for the
evaluation of the framework. Each mechanism is parameterized using the grid search algorithm\footnote{Also known
as parameter sweep.}~\cite{Lerman1980}.
The majority of masking mechanisms are parameterized with real numerical values.
A grid search discretizes these values, and then evaluates exhaustively all
possible combinations of parameter values.
\subsubsection{Laplace masking mechanism} 
This mechanism is widely used in literature \cite{Dwork2014, Aggarwal2008,
Duan2009}. The noise in the experiments of this  paper is generated by sampling a
laplace distribution with zero mean. The scale parameter \vmath{b} of the
distribution is selected to ensure maximum privacy. Part of privacy can be sacrificed to
increase utility if the privacy requirements from the \vDp{}s are not high. In this
masking mechanism, this is achieved by reducing the \vmath{b}. The scale
parameter for each laplace masking setting, is generated from value
\vmath{b=0.001} and during the parameter sweep the value increases by
\vmath{0.001} until it reaches \vmath{b=10}.

\subsubsection{Sine polyonym masking mechanism}\label{sec:sinePolyonym}
This mechanism is introduced in this paper. The mechanism generates random
noise that can be added to each sensor value. Assume a uniform random variable
\vUniformVar{}.
The noise generated from the introduced masking mechanism is calculated as
follows:
\begin{equation}
	\begin{gathered}
		\vMaskedValue{} = 
		\sum_{\xi=0}^{\absol{\Xi}}[\theta_{\xi}\sin(2\pi\vUniformVar)]^{2\xi+1}
	\end{gathered}
\end{equation}
The coefficients of the polyonym are denoted as \vmath{\theta_{\xi}}, and
\vmath{\xi} denotes the index of the coefficient. Both the length of the
polyonym \absol{\Xi} and the individual coefficient values can be tuned
to optimize the resulting privacy-utility values of the masking mechanism.
The generated noise is symmetrically distributed around zero, because the odd
power of the sine function produces both negative and positive noise with equal
probability. The sine function and its odd powers are always symmetrical towards
the horizontal axis, meaning that
\absol{\vmath{[\theta_{\xi}\sin(2\pi\vUniformVar)]^{2\xi+1}}} =
\absol{[-\theta_{\xi}\sin(2\pi\vUniformVar)]^{2\xi+1}} . Hence, the
integral of each factor is zero
\vmath{\int_{0}^{1.0}[\theta_{\xi}\sin(2\pi\vUniformVar)]^{2\xi+1}d\vUniformVar}
= 0. Therefore the distribution of generated values is symmetrical around zero
for \vmath{\vUniformVar \in \left[0,1\right]}, which denotes that the global
error mean is approximating zero.  \visibleComment{Increasing the length of the
polyonym and the values of its coefficients, increases the magnitude of the
local error, without affecting the global error, indicating that
higher utility can be achieved without sacrificing privacy. These properties
make polyonyms of trigonometric functions, such as sine and cosine, elligible
canditates for additive noise optimization.}
\vcomment{The majority of the additive noise mechanisms are based on symmetric
distributions tuned via a single scale variable.
In the proposed masking mechanism, the higher the number of polyonym factors the
more the parameters to tune the noise. Each factor of the polyonym adds a noise
that is scaled, based on the corresponding power and the coefficient value.}
By increasing the polyonym length and tuning the coefficient values, a larger
space of privacy settings is searched to maximize privacy and
utility.\\\indent
Each coefficient is assigned to a value in the space
\vmath{[0.01,1.8]}.
The grid search in that space starts with a step of \vmath{0.03} until the value of 0.3,
to evaluate settings that create low noise. Then the step changes to 0.3 until
the value of 1.8, to evaluate privacy settings that generate higher values of
noise. The sine polyonym masking settings are generated by creating all possible
permutations of these values for 5 coefficients. This yields around 10,000
masking settings. \visibleComment{Preliminary analysis
on the autocorrelation and the spectrograms of the proposed sine polyonym noise
does not show autocorrelation and recurring patterns over
different spectrograms\footnote{\visibleComment{Further analysis
on this, is possible future work and is out of the scope of this article. This can be
evaluated by introducing a metric that measures noise color in the privacy
function.}}.}

\subsection{Error analysis}\label{sec:errorAnalysis}
Each privacy setting that results from parameterization of the mechanisms is
evaluated by analyzing the local and the global error that they generate on
varying subset sizes of the ECBT dataset. By sampling varying sizes of the
dataset, the utility and privacy dispersion metrics are evaluated on a varying
number of sensor values, measuring the effect of varying participation
in the system. To create a random subset of the ECBT dataset, a
subset of users \vmath{\vUserSet{}_\text{test}} is chosen. In each repetition the users are
chosen randomly. All users use a universal privacy setting. The initial size
of the subset is 50 users, and then it increases by 50 users until \vmath{\absol{\vUserSet{}_\text{test}} = 500} users.
Then, the size of the subset increases by 500 users until
\vmath{\absol{\vUserSet{}_\text{test}} = 6,435}. This process generates several
local and global error values. The average, standard deviation and entropy of
the local error and global error are calculated for all samples generated
from the above process. The empirical cumulative distribution
function\footnote{The cumulative distribution function denotes the probability of a generated value
being lower or equal than the corresponding domain axis
value~\cite{Spiegel1992}.} (CDF) is shown for each metric in
\reffig{fig:stats:error}.

The sine polyonym mechanism can produce a wider range of local and global error
values compared to the laplace mechanism, since almost every sine polyonym CDF
curve is covering a wider domain range on the domain axis compared to the
respective laplace CDF curves.
The majority of the range axis values of the sine polyonym CDF curve are higher
than the corresponding range values of the laplace CDF curve. This indicates
that it is more probable to generate lower global or local error value by using a sine
polyonym setting compared to a laplace setting. Concluding, the sine polyonym
settings are expected to produce a wider range of privacy-utility trade-offs. Based on
the CDF charts, sine polyonym settings are more likely to achieve higher
utility, whereas laplace settings are expected to achieve higher privacy.

\begin{figure}[ht!]
  \subfloat[Local error mean]{%
       \includegraphics[width=0.45\linewidth]{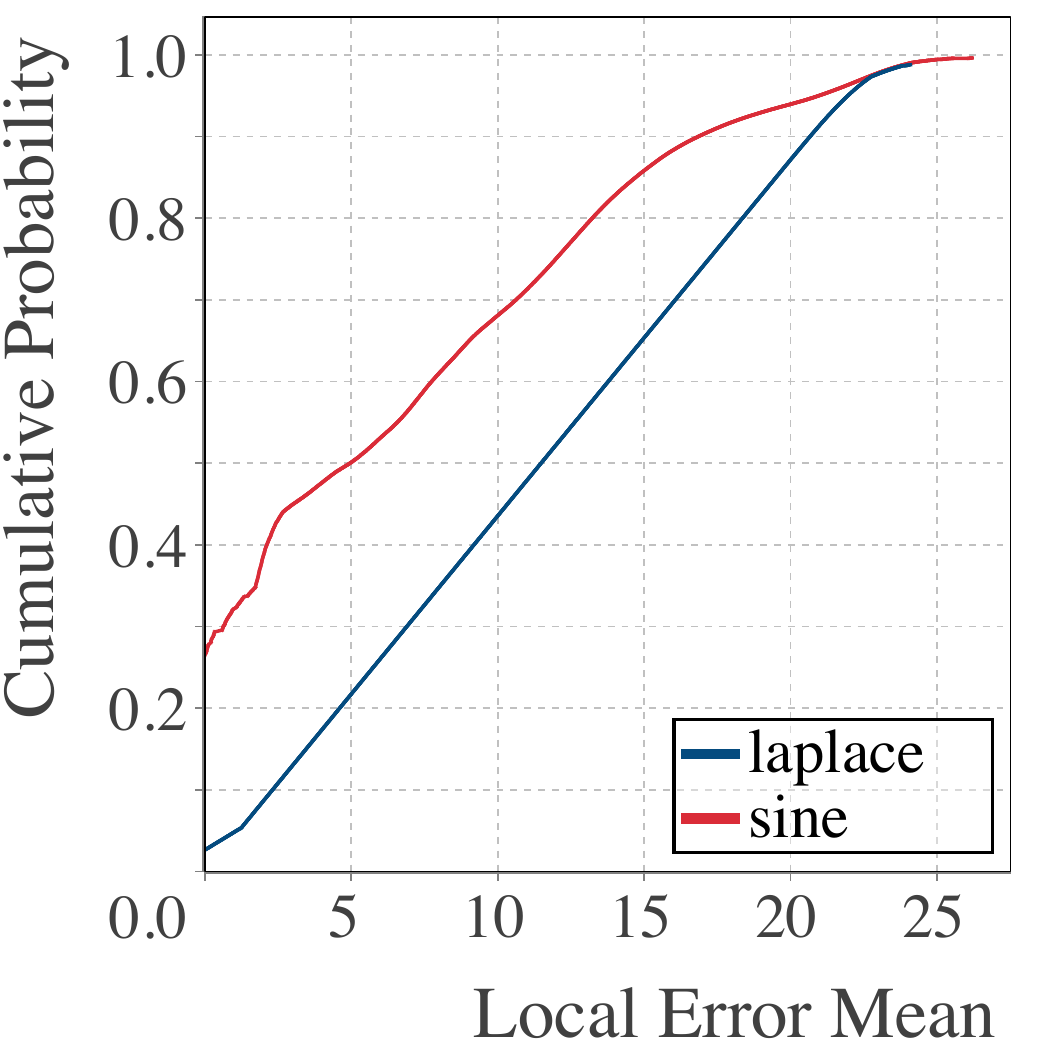}
       \label{fig:cdf:a}}\hfill
  \subfloat[Global error mean]{%
        \includegraphics[width=0.45\linewidth]{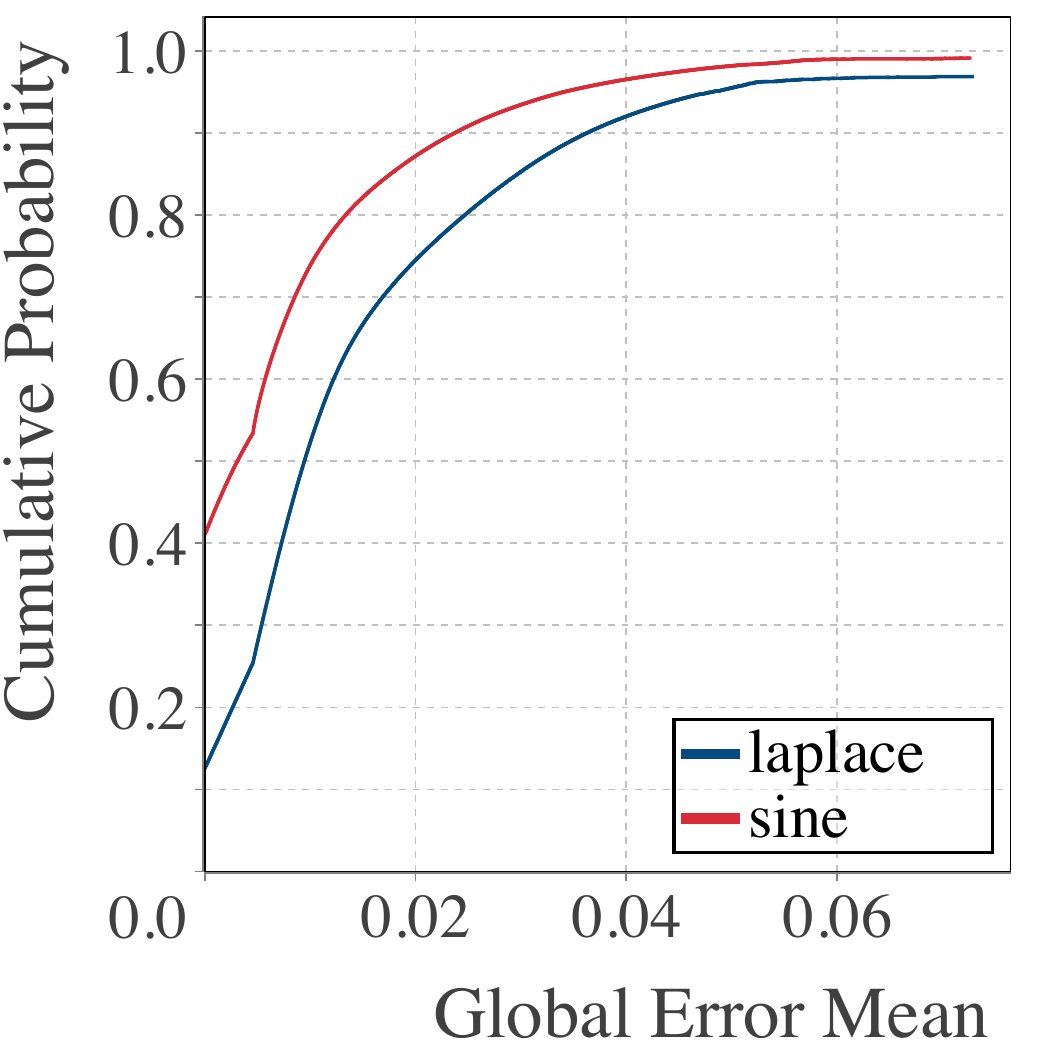}
        \label{fig:cdf:b}}\\
   \subfloat[Standard deviation of local error]{%
       \includegraphics[width=0.45\linewidth]{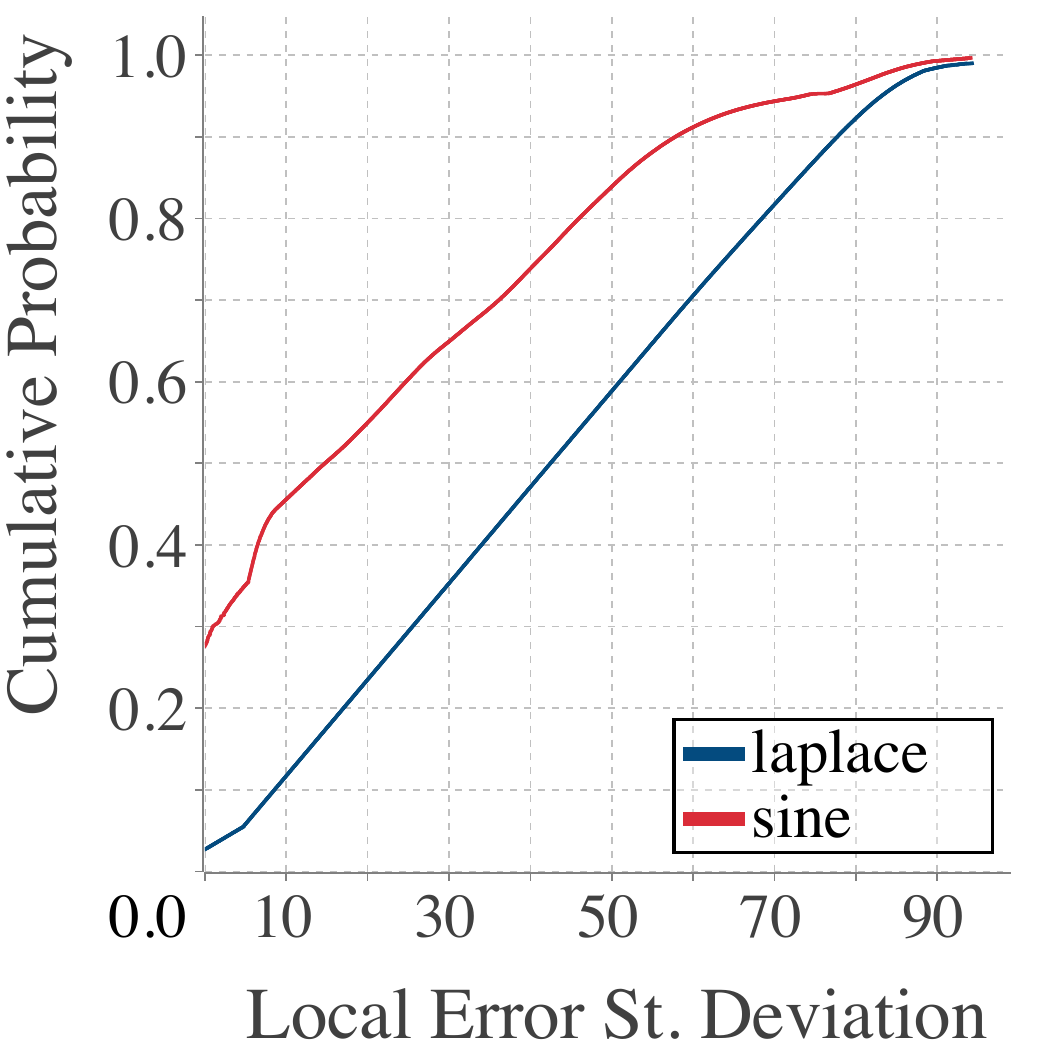}
       \label{fig:cdf:c}}\hfill
  \subfloat[Standard deviation of global error]{%
        \includegraphics[width=0.45\linewidth]{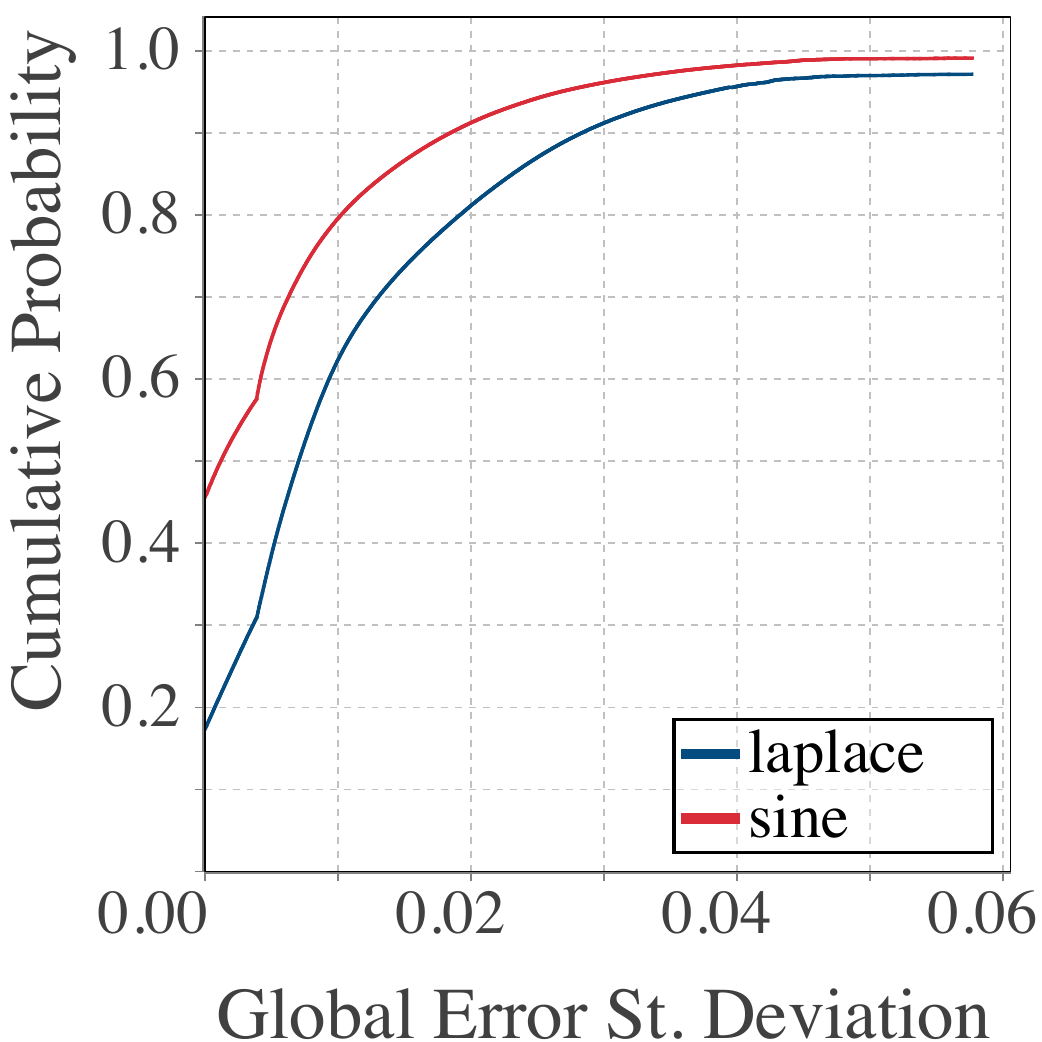}
        \label{fig:cdf:d}}\\
  \subfloat[Local error entropy]{%
       \includegraphics[width=0.45\linewidth]{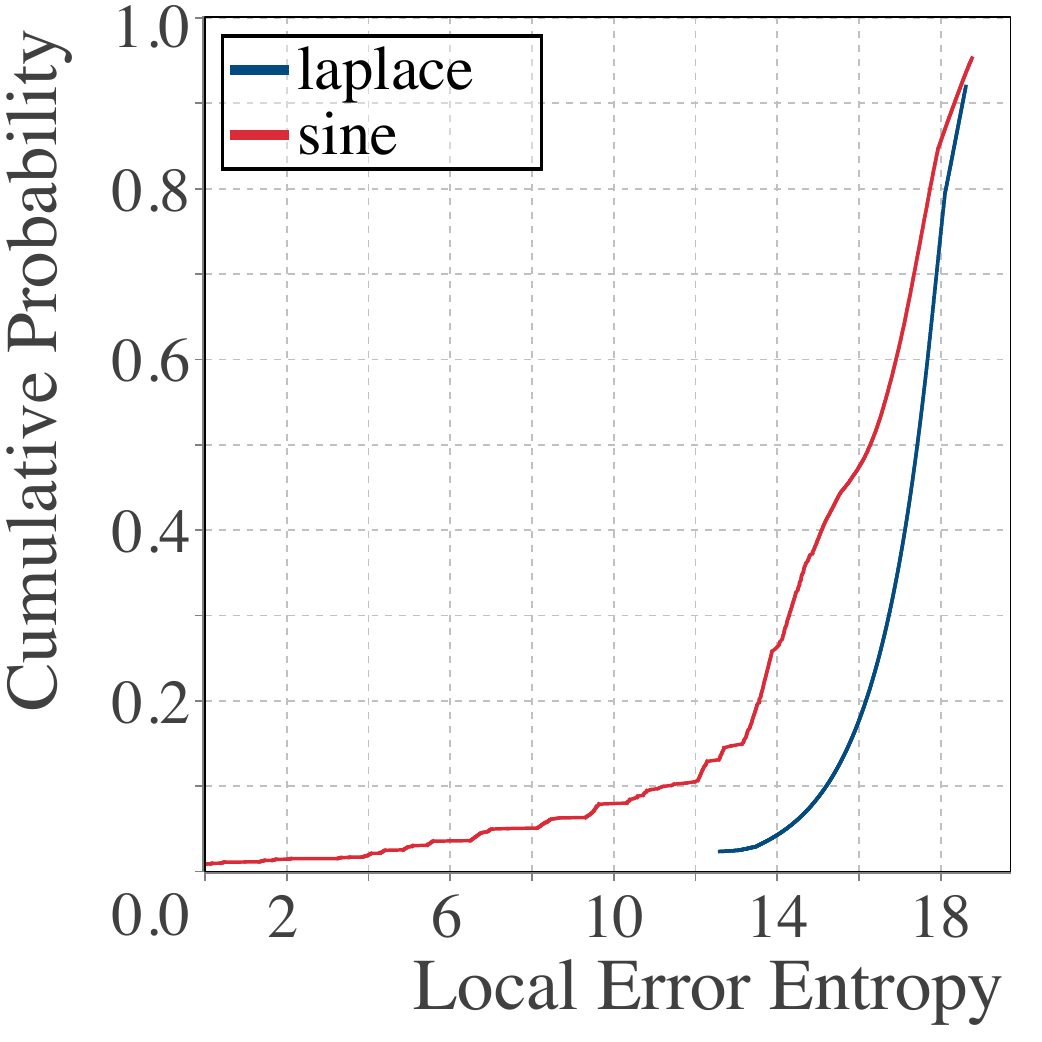}
       \label{fig:cdf:e}}\hfill
  \subfloat[Global error entropy]{%
        \includegraphics[width=0.45\linewidth]{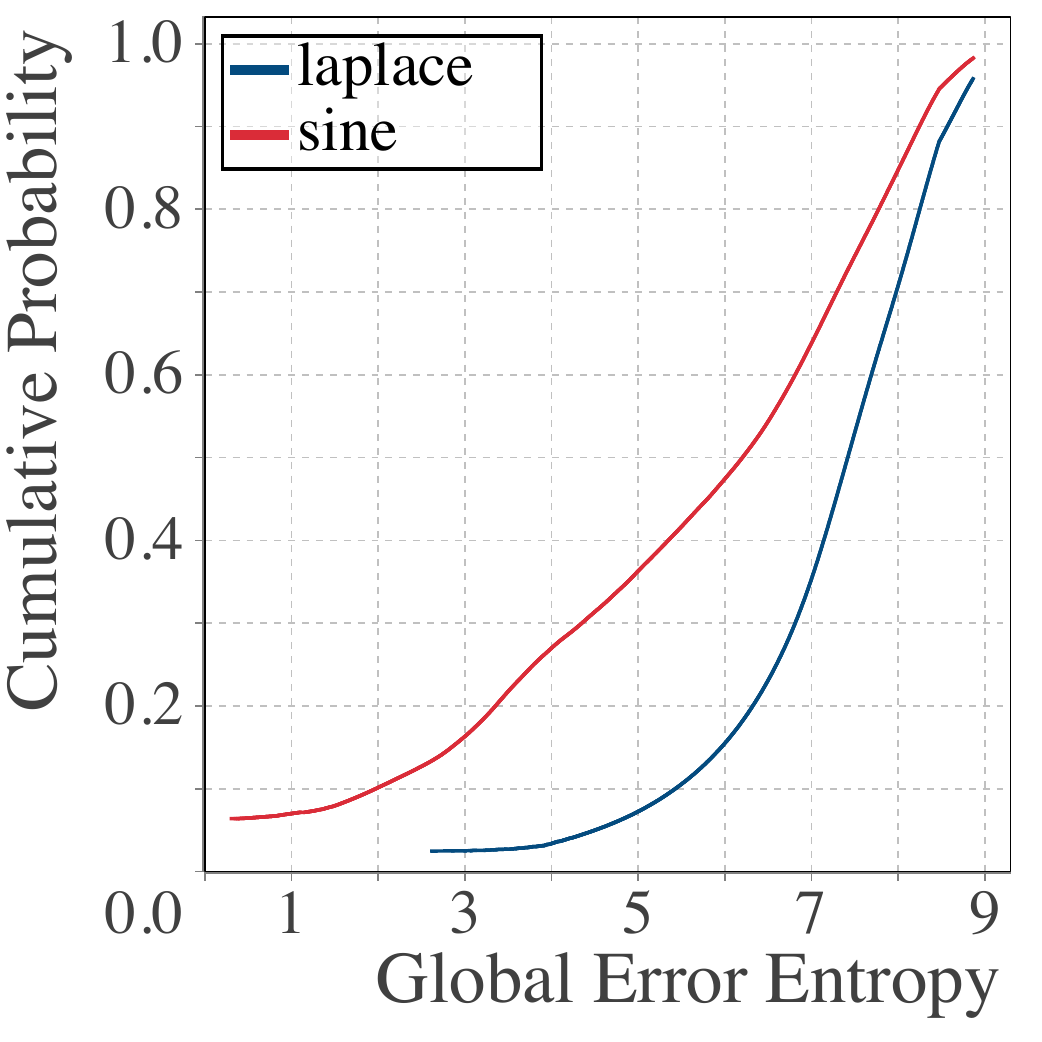}
        \label{fig:cdf:f}}\\
  \caption{Cumulative distribution function of each
  local and global error metric computed by all settings of each masking
  mechanism.
  }
  \label{fig:stats:error}
  \vspace{-0.27in}

\end{figure}

\subsection{Parameter analysis}
For the experiments, \vmath{\alpha} and \vmath{\gamma} parameters are
defined to calculate the privacy and utility. The choice of these parameters may
vary based on the distribution of the sensor values and the kind of
aggregation.
Also \vDp{}s and \vDc{}s may have varying requirements that affect the choice of
those values. In this paper, these values are determined empirically, to showcase an
empirical evaluation. If a \vDc{} successfully calculates the local error mean
by acquiring the corresponding original values of a masked set, then it is
possible to estimate the original sensor values of other masked sets as well, by
subtracting the calculated mean.
This challenge is addressed by using privacy settings with high noise variance.
Still, high variance does not guarantee that the masking process is not irreversible.
If noise varies between a small finite number of real values, then the \vDc{}
can also estimate the original value of the data by subtracting the variance. To overcome this challenge, privacy settings that produce noise with
high entropy, therefore high randomness, are chosen.
Consequently, a lower value for the coefficient of local error mean is chosen
as \vmath{\alpha_1 = 0.2}, while entropy and standard deviation of the local error
share a higher coefficient value of \vmath{\alpha_2 = \alpha_3 = 0.4}.

Assigning values to the utility coefficients depends highly on the
preferences of the \vDc{}. In the case of sum, the global error mean should be near 0, unless
the \vDc{} estimates the mean and then subtracts it from the aggregation
result. For this paper the main concern is to keep a global error mean near
zero, to avoid the aforementioned correction process. Standard deviation and
entropy are assigned with equal weight. Therefore, a very high
coefficient of \vmath {\gamma_1 = 0.6} for the global error mean is chosen,
whereas the coefficients of \vmath{\gamma_2=\gamma_3=0.2} for global error,
standard deviation and entropy are chosen. To avoid 
evaluating mechanisms with high utility dispersion and low utility mode
values, a hard constraint is applied and only mechanisms that generate mean
\vmath{\vMean{\vGlobalSet} < 0.1} and standard deviation values
\vmath{\vStandardDeviation{\vGlobalSet} < 0.1} are evaluated.
The normalizing factors of Relations (\ref{eq:utility}) and (\ref{eq:privacy})
are chosen after the application of this constraint.

A sensitivity analysis of the parameters for each masking mechanism is
performed to evaluate the effect of different parameter values on the privacy
and utility output of each masking mechanism. In the laplace masking mechanism,
increasing the scale parameter \vmath{b} of the distribution, also increases
the total noise added to the dataset. In the sine mechanism, increasing the number
and values of the coefficients, also increases the total generated noise. In
\reffig{fig:comparison}, a comparison of privacy and
utility is shown between the two types of mechanisms. The values of utility and
privacy are generated as shown in \refsec{sec:errorAnalysis}. The total
noise is generated by measuring the noise level of each privacy setting on a sample of
100,000 sensor values\footnote{This sample size is chosen to be large enough for statistical
significance and small enough to reduce computation costs.}.
The lines are smoothed by applying a moving average, to make the comparison
clearer. For the same amount of total absolute generated noise
\vmath{\sum_t\absol{\vNoiseValue_t}}, the laplace privacy settings achieve
higher privacy, often more than \vmath{1\%} over the sine polyonym privacy
settings. The sine polyonym privacy settings achieve higher utility around
\vmath{1\%} over the laplace privacy settings. Therefore the results illustrated
in \reffig{fig:stats:error} are reflected in the privacy and utility values
generated from the above parameterization. Moreover, the trade-off between
privacy and utility is observable, as privacy increases with the decrease of
utility and vice versa for both mechanisms.

\begin{figure}[htp]
\begin{center}
  \includegraphics[width=0.7\linewidth]{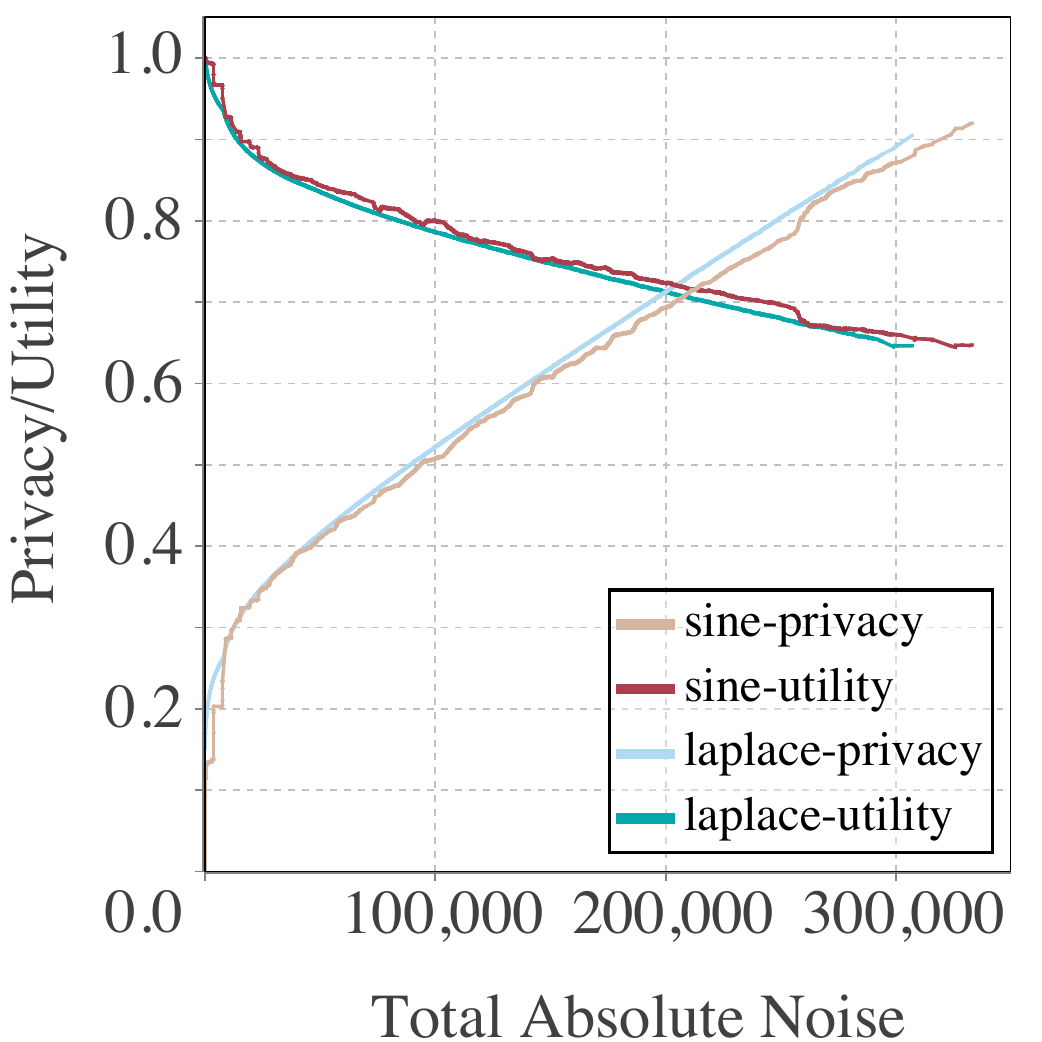}
  \caption[labelInTOC]{Comparison of sine polyonym and laplace masking
  mechanisms in terms of privacy and utility.}
  \label{fig:comparison}
\end{center}
\vspace{-0.27in}
\end{figure}

\subsection{Homogeneous system evaluation}\label{sec:hom}
All the generated privacy settings are evaluated via the framework proposed in
\refsec{sec:framework}. The proposed framework filters out five privacy
settings for five privacy bins of size 0.2. The constraint value for evaluating
privacy settings is chosen empirically to be half of the bin size \vmath{\omega
= 0.1}, to ensure low privacy dispersion, based on Relation
(\ref{eq:ptivcon}). The resulting privacy settings are summarized in
\textit{Table} \ref{tab:results}. The last two columns of the table, illustrate
the median privacy and utility values for each masking mechanism. The first
column shows the id of each setting, which is used as reference in Figures
\textit{\ref{fig:optim} \& \ref{fig:heat}}.

\begin{table}
	\caption{A table summarizing the performance of the five optimal privacy
	settings based on the parameters of the sine polyonym denoting the coefficient
	value for each factor of the polyonym or the scale value for a laplace
	mechanism. In case of the sine polyonym, the first number from right is mapped
	to the first factor (\vmath{\xi=1}) and so on.}\label{tab:results}
	\footnotesize
	\begin{tabular}{ c| l| l| c| c }
	\multicolumn{1}{c}{ID} &
	\multicolumn{1}{c}{Masking} & \multicolumn{1}{c}{
	Parameters} & \multicolumn{1}{c}{Privacy} &
	\multicolumn{1}{c}{Utility} \\ \hline
	A & cosine & 0.0-0.0-0.0-0.18-0.0	& 0.01 &	0.99\\
	B & laplace & 0.005					& 0.20	& 0.98\\
	C & cosine & 0.6-0.6-0.0-0.9-0.3	& 0.40	& 0.84\\
	D & cosine & 1.2-0.3-0.6-1.2-0.9	& 0.60	& 0.76\\
	E & cosine & 1.5-1.5-1.2-0.3-1.2	& 0.80 &	0.68\\
	N & none &	-	&	0.00	&	1.00\\

	\end{tabular}
	  \vspace{-0.27in}
\end{table}

\reffig{fig:optim:all} shows the generated privacy-utility values for all the
privacy settings tested. Each color is mapped to the masking mechanism that is
used to produce this setting. The line denotes the median value of utility at
the given privacy value. The non-median privacy-utility values occur in the
semi-transparent area. Upper and lower edges of the area denote the minimum
and maximum utility value for the corresponding privacy value. Lower utility
values for a given privacy point are generated from applications of the privacy
setting on small subsets of the ECBT datasets, where \vmath{\absol{\vUserSet}
\leq 1000}. The number of sensor values decreases with the number of users.
\vcomment{Consequently, the number of noise values decreases as well.} Therefore, the noise
cancellation is also reduced, as mentioned in \refsec{sec:per}. Hence, subsets
with a lower number of sensor values produce lower utility values. The trade-off
between privacy and utility is quantified, since the median
curve and the edges of the surrounding area indicate a decrease in utility with
the increase of privacy.
In \reffig{fig:optim:selected}, the area of privacy-utility values of 5 privacy
settings produced by the optimization process is shown in
\refsec{sec:utilEval}. Furthermore the ``no masking'' privacy setting is also
considered, where users choose to use no privacy setting and send the values
unmasked.

As it is shown, the privacy values of each privacy setting are within a range
of lower than 0.2 privacy. The dispersion of utility increases for privacy
settings that achieve higher privacy. The importance of offering more rewards for the usage of
higher utility mechanisms is validated, since high dispersion of utility is
restrictive for accurate sum calculations by the \vDc{}s. Figures
\ref{fig:optim:all:lim} and \ref{fig:optim:selected:lim} illustrate the
privacy-utility trajectories for more than 1,000
users. It is evident that a \vDc{} can also increase utility and reduce its
dispersion by attracting more users.
Higher rewards in general, can also attract more users, so the utility
dispersion is expected to decrease even more.

\begin{figure}[ht!]
  \subfloat[Trajectory for all user set sizes]{%
       \includegraphics[width=0.48\linewidth]{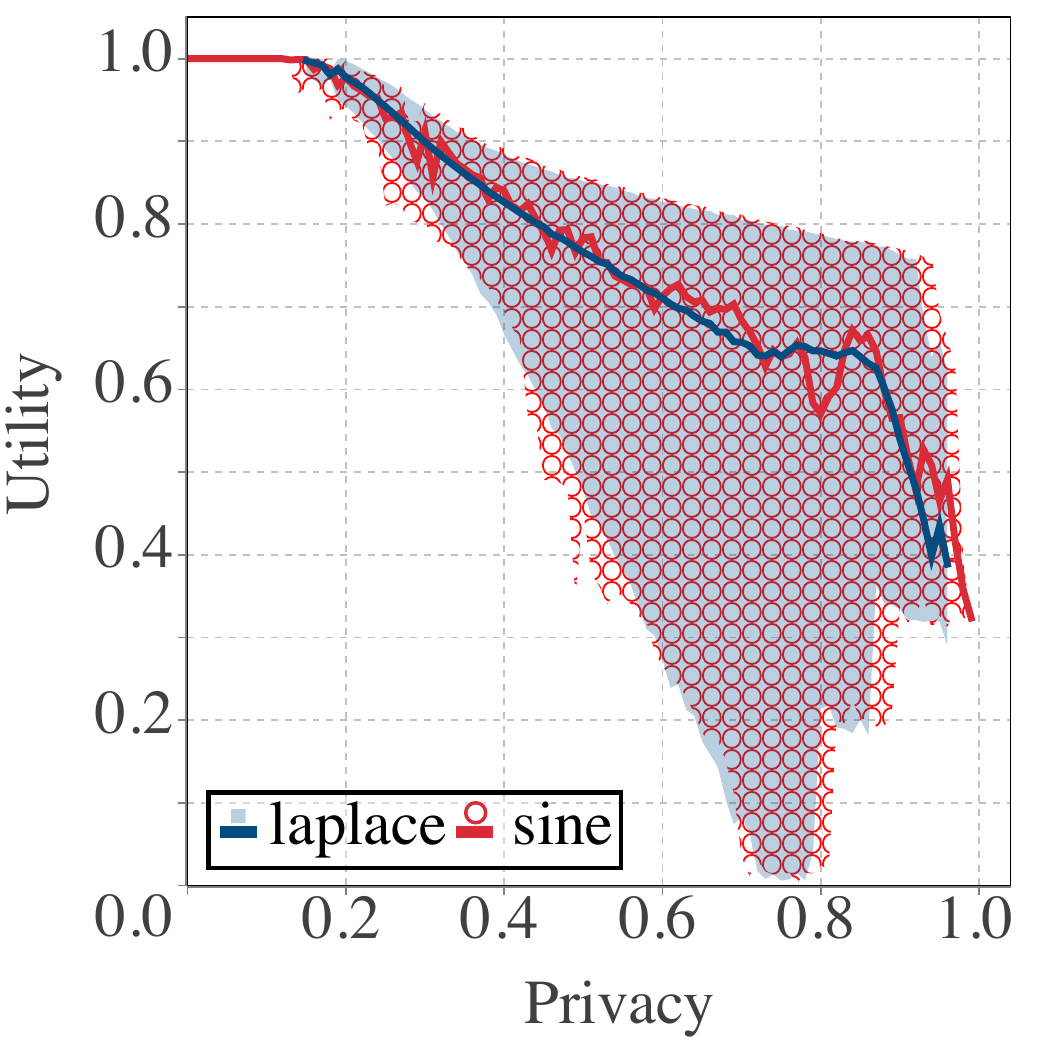}
       \label{fig:optim:all}}\hfill
  \subfloat[Optimization results for all user set sizes]{%
        \includegraphics[width=0.48\linewidth]{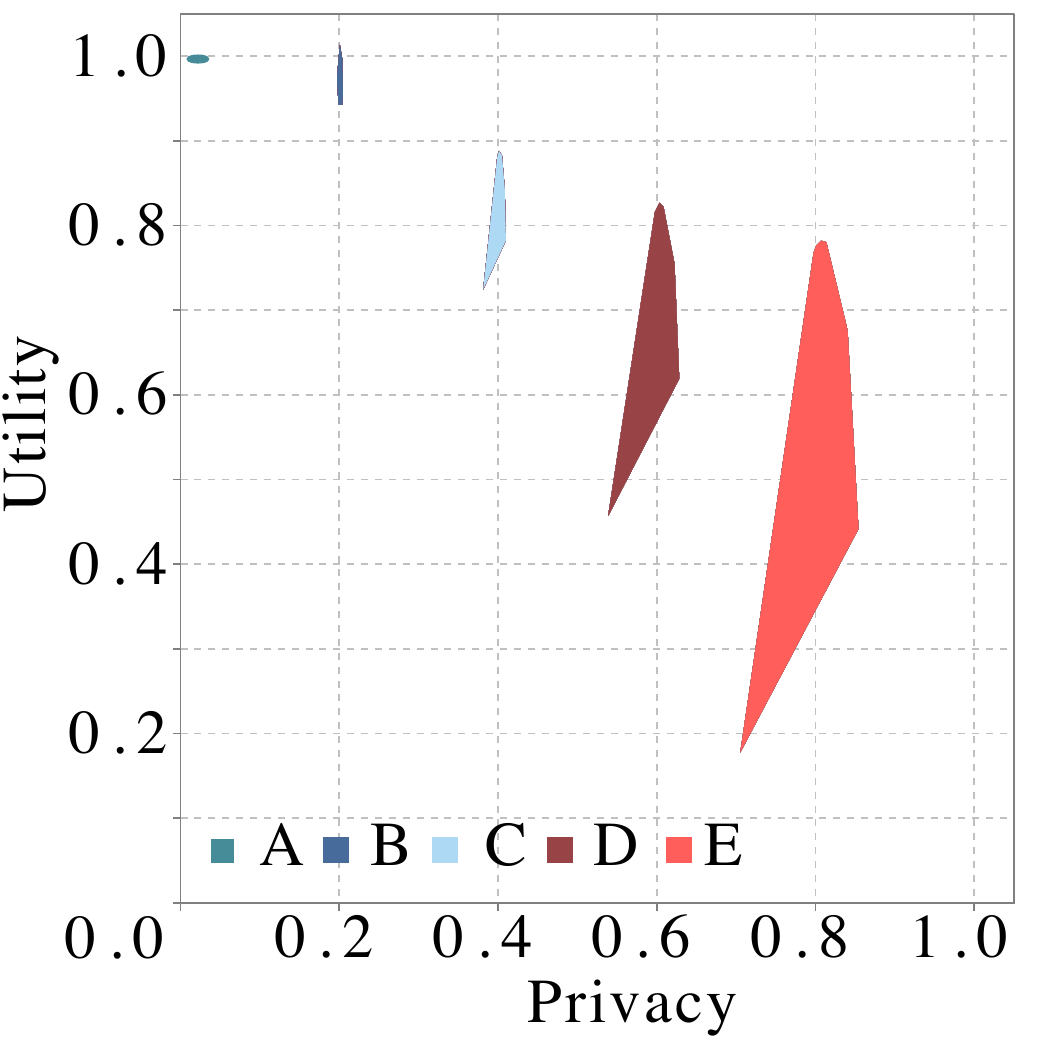}
        \label{fig:optim:selected}}\hfill\\
   \subfloat[Trajectory for all user sets with more than 1000 users]{%
       \includegraphics[width=0.48\linewidth]{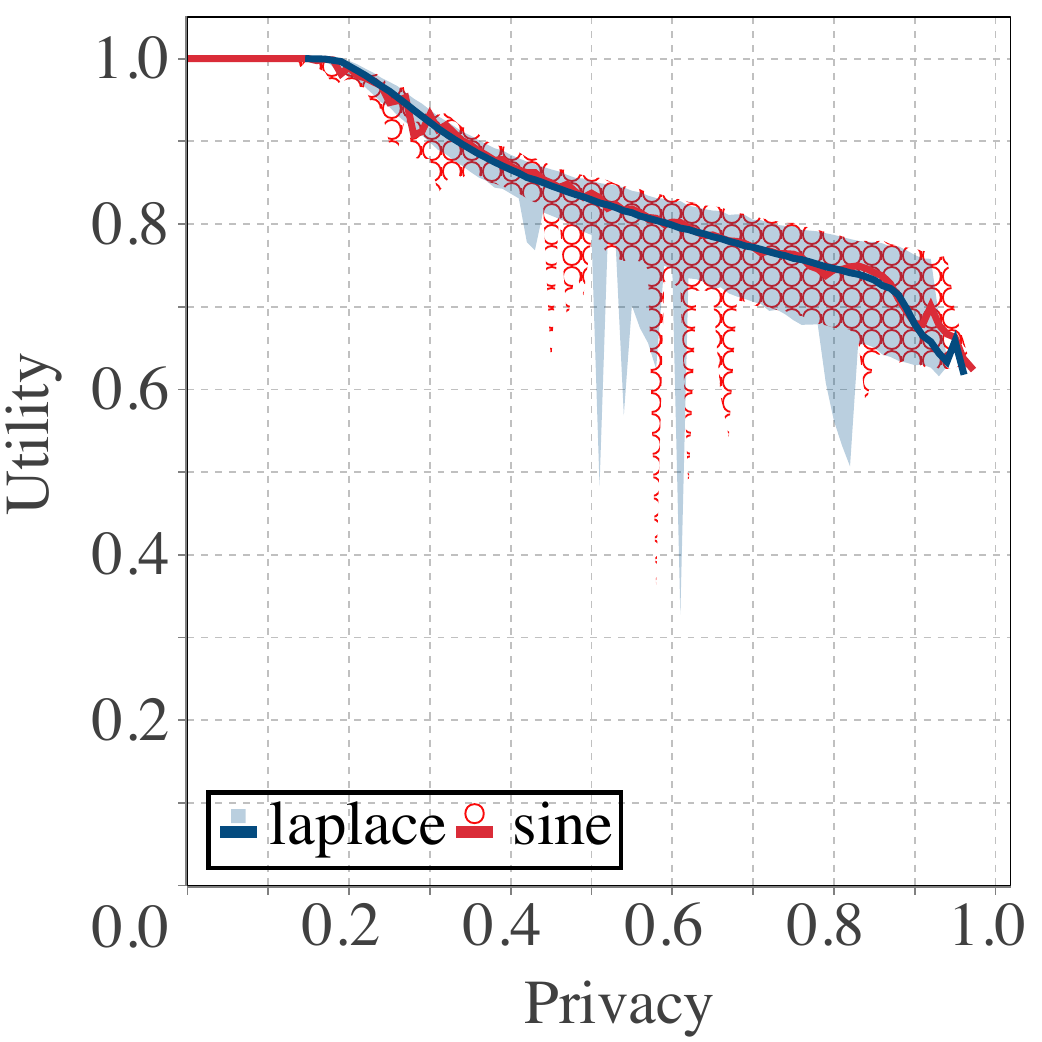}
       \label{fig:optim:all:lim}}\hfill
  \subfloat[Optimization results for user sets with more than 1000 users]{%
        \includegraphics[width=0.48\linewidth]{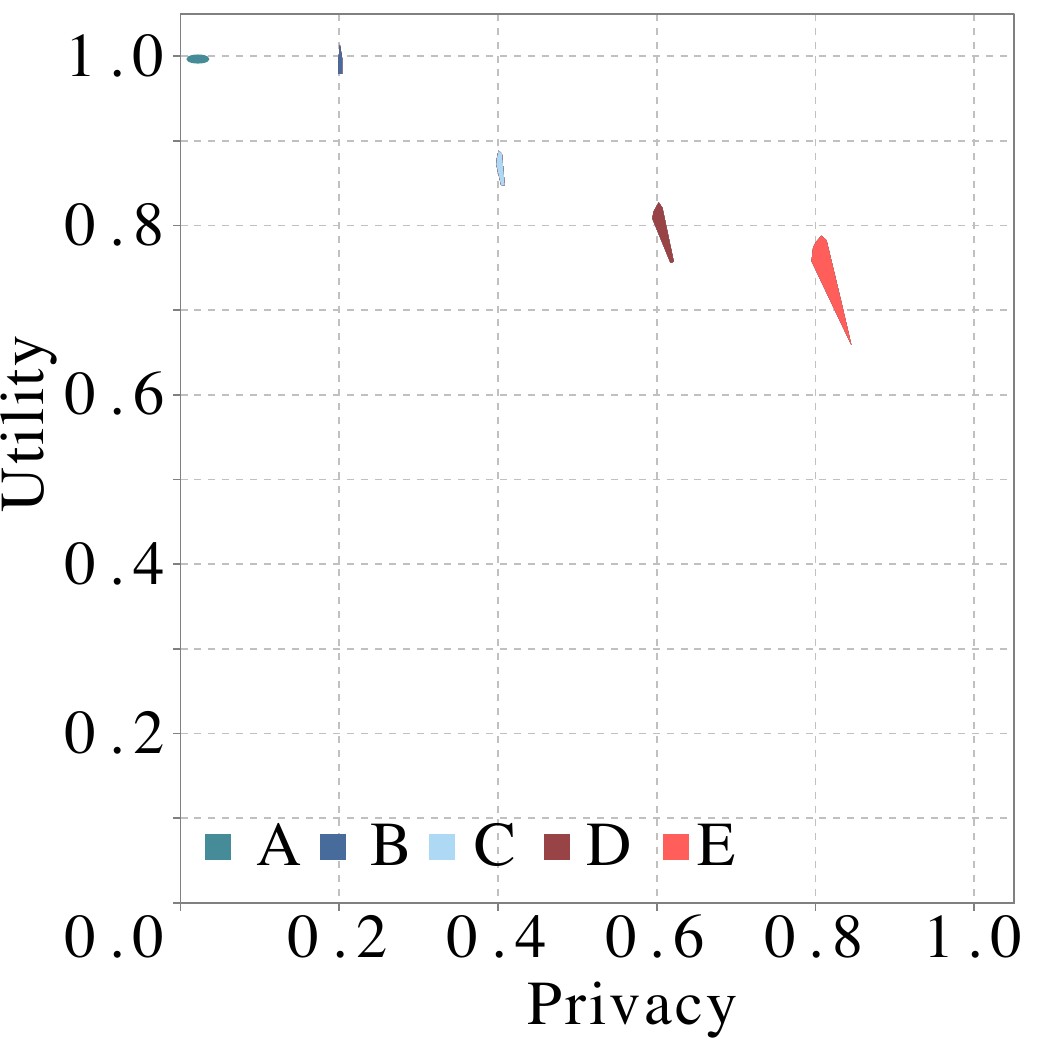}
        \label{fig:optim:selected:lim}}\hfill\\
  \caption{Figures \ref{fig:optim:all} \& \ref{fig:optim:all:lim} show the
  privacy-utility trajectory of the privacy settings grouped by masking
  mechanisms in the same color. Figures \ref{fig:optim:selected} \&
  \ref{fig:optim:selected:lim} illustrate the trajectories of the privacy
  settings, which are generated by the proposed framework.}
  \label{fig:optim} 	
  \vspace{-0.2in}
\end{figure}

\subsection{Heterogeneous system evaluation}
In an heterogeneous system, the framework performance is evaluated under the use of
different privacy settings from each user. The difference of privacy and utility
between homogeneous and heterogeneous systems is quantified.
This quantification is done by performing an exhaustive simulation.
The simulation combines the ECBT dataset and the six privacy settings in
\textit{Table} \ref{tab:results}. Every user of the ECBT dataset is assigned a
privacy setting from \textit{Table} \ref{tab:results}. The percentage of users
that are assigned each privacy setting is parameterized A histogram with six
bins is created. Each bin corresponds to the ID of a specific privacy setting
from Table \ref{tab:results}. The percentage assigned to a bin denotes the
percentage of users using the respective privacy setting at this time point.
To generate several possible scenarios for different distributions of user
choices, the histogram is parameterized via a parameter sweep of all possible
percentage values for each setting, with a step of \vmath{12.5\%}. This process
produces over 1000 possible histograms. In figures \ref{fig:heat:priv:center} -
\ref{fig:heat:ut:spread} the heatmaps show the median and the interquantile
range (IQR)\footnote{IQR is considered a robust measure of scale, which is
especially used for non-symmetric distributions. It measures the range between
the \vmath{25^{th}} and the \vmath{75^{th}} quantiles.} of privacy and utility
for all histograms that the privacy setting has a higher percentage of users
compared to the others. Such a setting is referred to as dominant setting. This
sorting of settings is done to examine the privacy-utility changes while users
move from a higher to the next lower utility setting. The top row of the heatmap
shows the homogeneous scenario case, where 100\% of the users chose only one
setting.

The analysis of the heatmap in \reffig{fig:heat:priv:center} shows an increase
in privacy when the majority of users choose the more privacy-preserving
settings of the homogeneous scenario. This effect is observed for any percentage
of users for a dominant setting. A decrease in utility median is confirmed in
\ref{fig:heat:ut:center}, when the majority of users shifts from less private to
more private settings. The trade-off between privacy and utility is preserved in
the heterogeneous scenario, regardless of the percentage of users that choose
the dominant setting. Privacy values disperse more in heterogeneous systems,
according to \reffig{fig:heat:priv:spread}, as the percentage assigned to the
dominant setting drops. The dispersion of privacy can reach up to 0.16, which is
still lower than the bin range. In terms of utility, the dispersion is much
lower on average. There is a dispersion of around 0.1
for high utility mechanisms when they are dominant with 87.5\% of users. A
possible explanation for this is the reduction of noise cancellation of high
privacy settings, due to the low percentage of users choosing them.
Concluding, changing from a homogeneous system to a heterogeneous system
preserves the trade-off between privacy and utility in the median values. Furthermore,
the change to a heterogeneous system increases the dispersion of privacy-utility
values for all the mechanisms, so the \vDc{} should expect the aggregates to be less accurate.
Still, utility remains over 0.76 even if the IQR is subtracted from the median,
indicating that the aggregate is still approximated even in the heterogeneous
case. This validates empirically Theorem \ref{thm:heterogeneous}. In both
cases it is efficient for the \vDc{} to shift user privacy choices to high
utility mechanisms by offering them higher rewards. The randomness of the
generated noise in an heterogeneous system does not create high variance or high expected
global errors. Individual privacy is still preserved for
all users and their privacy settings. The individual privacy value does not
change between heterogeneous and homogeneous systems, since the privacy-setting
choice of one user does not affect the added noise to the sensor values of the
other ones.

\begin{figure}[ht!]
  \subfloat[Privacy Median]{%
       \includegraphics[width=0.48\linewidth]{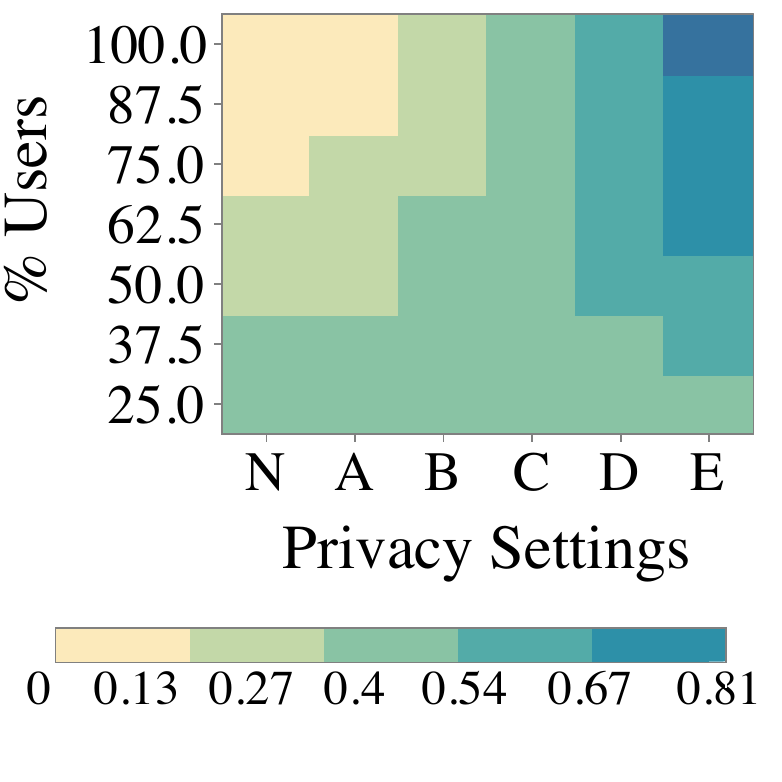}
       \label{fig:heat:priv:center}}\hfill
  \subfloat[Privacy IQR]{%
        \includegraphics[width=0.48\linewidth]{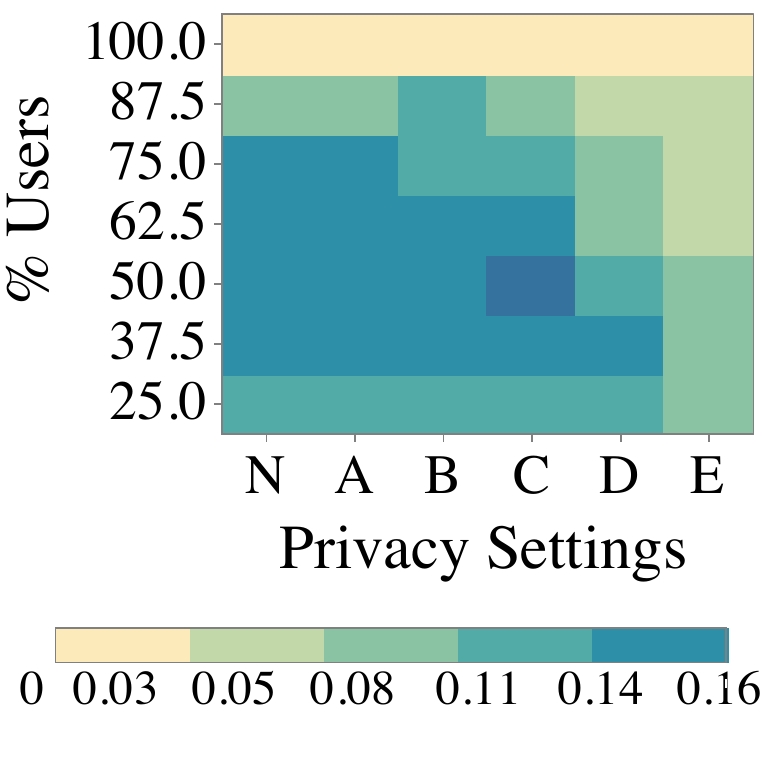}
        \label{fig:heat:priv:spread}}\hfill\\
   \subfloat[Utility Median]{%
       \includegraphics[width=0.48\linewidth]{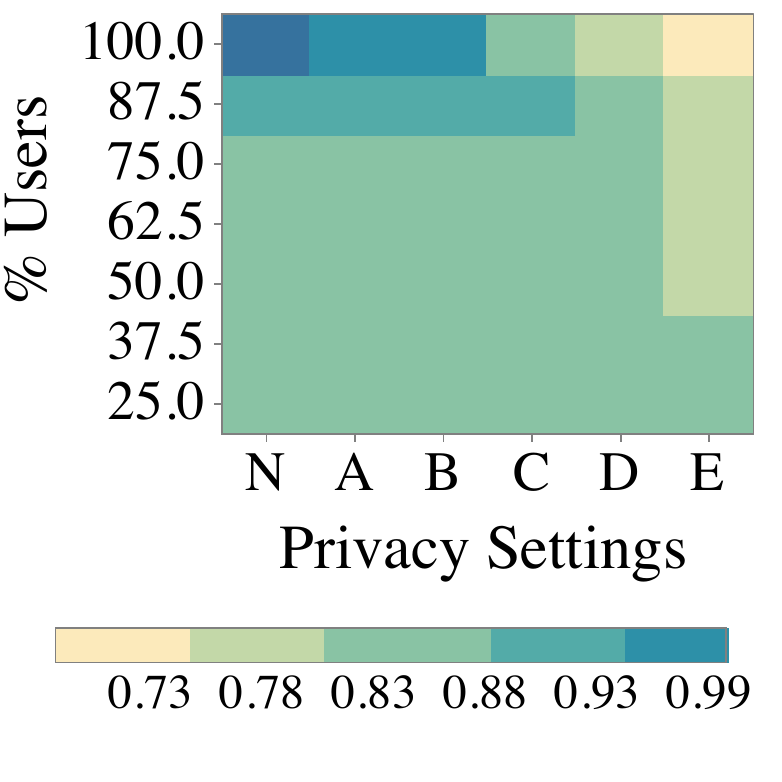}
       \label{fig:heat:ut:center}}\hfill
  \subfloat[Utility IQR]{%
        \includegraphics[width=0.48\linewidth]{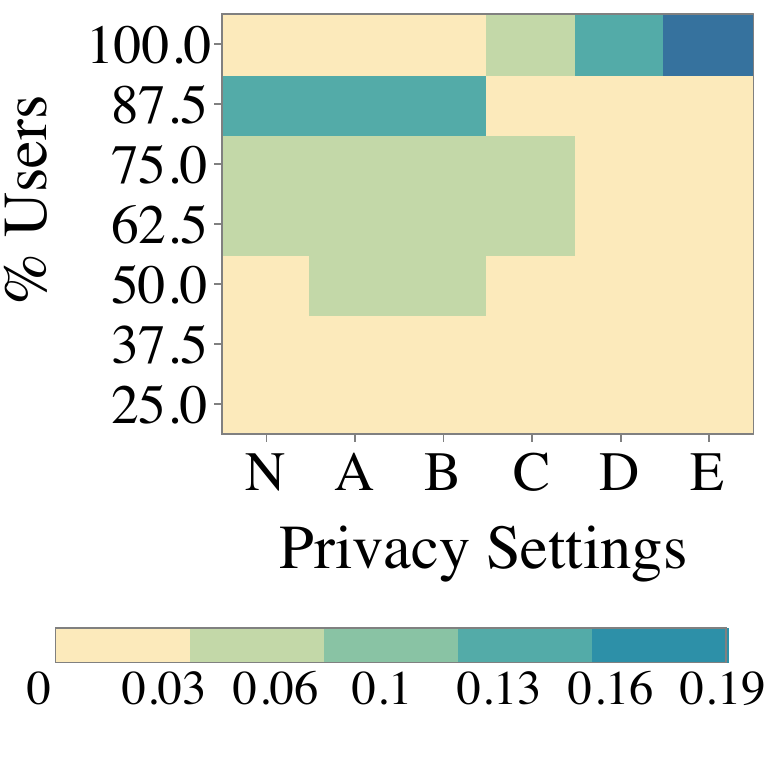}
        \label{fig:heat:ut:spread}}\hfill\\
  \caption{The heatmaps in Figures
  \ref{fig:heat:priv:center}-\ref{fig:heat:ut:spread} show the
  privacy and utility median and IQR values, for various distributions of
  privacy setting choices among the users.}
  \label{fig:heat} 	
    \vspace{-0.27in}
\end{figure}


\section{Conclusion and Future Work}\label{sec:conclusion}

An optimization framework for the selection of privacy settings is introduced in
this paper. The framework computes privacy settings that maximize utility for
different values of privacy. This framework can be utilized in
privacy-preserving systems that calculate aggregation functions over
privatized sensor data. The \vDp{}s of such system can self-determine the
privacy setting of their choice, since it is guaranteed that it produces the
desired privacy with very low deviation. For the \vDc{} of the system, it is
guaranteed that if the \vDp{}s are incentivized to use low-privacy settings and
high utility settings, the approximated aggregate is highly accurate. Analytical
as well as empirical evaluation using over \vmath{20,000} privacy settings and
real-world data from a Smart Grid pilot project confirm the viability of
participatory data sharing under informational self-determination.

For future work, the proposed framework can be improved by incorporating a
machine learning process that computes personalized recommendations of privacy
settings to each \vDp{}, by identifying the prior distribution of the sensor
data and also the preferences of the \vDp{}.  Further empirical
evaluations of framework can be performed by implementing other aggregation
functions and using different datasets.  \visibleComment{Finally, an analytical
proof that the sine polyonym additive noise is not colored and differentially
private can be performed.}

\section{Acknowledgments}
This work is supported by European Community’s H2020 Program under the scheme
'ICT-10-2015 RIA', grant agreement \#688364 
‘ASSET: Instant Gratification for Collective Awareness and Sustainable
Consumerism’\footnote{\url{http://www.asset-consumerism.eu/}}.




\section{References}



\begin{thenomenclature} 

 \nomgroup{A}

  \item [{$A$}]\begingroup A multiset of real values. Any capital letter is treated as a multiset of real values, unless stated otherwise.\nomeqref {19}
    \nompageref{16}

 \nomgroup{B}

  \item [{\vAggregate{A}}]\begingroup a function that aggregates all elements of a set \vmath{A} into a real value. e.g. for sum: \vmath{g_{sum}\left(A\right) = \sum_{i=0}^{\absol{A}}a_i}\nomeqref {19}
    \nompageref{16}
  \item [{$\vMean{A}$}]\begingroup The mean value of all elements of a set, where \vmath{a_{i} \in A}.\nomeqref {19}
    \nompageref{16}
  \item [{\vMedian{A}}]\begingroup The median value of all elements of a set, where \vmath{a_{i} \in A}. \nomeqref {19}
    \nompageref{16}
  \item [{\vMax{A}}]\begingroup The maximum value of all elements of a set, where \vmath{a_{i} \in A}. \nomeqref {19}
    \nompageref{16}
  \item [{\vMin{A}}]\begingroup The minimum value of all elements of a set, where \vmath{a_{i} \in A}. \nomeqref {19}
    \nompageref{16}
  \item [{\vEntropy{A}}]\begingroup The Shanon's entropy value for all elements of a set, where \vmath{a_{i} \in A}. \nomeqref {19}
    \nompageref{16}

 \nomgroup{C}

  \item [{\subopt{a}}]\begingroup A suboptimal value that approaches an optimal value, e.g. \vmath{\subopt{a}\,\to\,\vMax{A}} or \vmath{\subopt{a}\,\to\,\vMin{A}}. \nomeqref {19}
    \nompageref{16}
  \item [{\suboptrelax{a}}]\begingroup A new suboptimal value that approaches an existing suboptimal value \subopt{a}.\nomeqref {19}
    \nompageref{16}
  \item [{\vUser{}}]\begingroup A user\nomeqref {19}\nompageref{16}
  \item [{\vUserSet{}}]\begingroup A set of users \nomeqref {19}
    \nompageref{16}
  \item [{\vTime{}}]\begingroup A time index\nomeqref {19}
    \nompageref{16}
  \item [{\vSensorValueIndexed{}}]\begingroup  A sensor value generated in time \vTime by the user \vUser\nomeqref {19}
    \nompageref{16}

 \nomgroup{D}

  \item [{\vMaskingId{}}]\begingroup a masking mechanism, which consists of a parametric algorithm that masks the sensor values of a multiset \vmath{\vSensorValueSet}.\nomeqref {19}
    \nompageref{16}
  \item [{\vMaskingParametersIndexed}]\begingroup A parameterization \vParamId for a masking mechanism \vMaskingId.\nomeqref {19}
    \nompageref{16}
  \item [{\vUniformVar}]\begingroup A uniformly distributed variable.\nomeqref {19}
    \nompageref{16}
  \item [{\vPrivacySetting{\vSensorValueSet}}]\begingroup a privacy setting consisting of a masking mechanism \vMaskingId parameterized with parameters \vMaskingParametersIndexed and operating on a set of sensor values \vSensorValueSet. It produces a masked set of sensor values \vmath{\vPrivacySetting{\vSensorValueSet} = \vMaskedSet}, such that \vmath{\absol{\vSensorValueSet} = \absol{\vmath{\vMaskedSet}}}.\nomeqref {19}
    \nompageref{16}
  \item [{\vPrivacySet}]\begingroup A multiset of privacy values.\nomeqref {19}
    \nompageref{16}
  \item [{\vmath{\alpha_i}}]\begingroup A parameter that weights the importance of privacy factors for calculating the privacy values.\nomeqref {19}
    \nompageref{16}
  \item [{\vmath{\delta}}]\begingroup A parameter that denotes the amount of privacy that the \vDp{} sacrifices or gains over the existing privacy.\nomeqref {19}
    \nompageref{16}
  \item [{\vmath{c}}]\begingroup A parameter that denotes the amount of utility that a \vDc{} sacrifices or gains over the existing utility value.\nomeqref {19}
    \nompageref{16}
  \item [{\vUtilitySet}]\begingroup A multiset of utility values.\nomeqref {19}
    \nompageref{16}
  \item [{\vmath{\alpha_i}}]\begingroup A parameter that weights the importance of privacy factors for calculating the utility values.\nomeqref {19}
    \nompageref{16}
  \item [{\vmath{\gamma_i}}]\begingroup A parameter that weights the importance of utility factors for calculating the utility values.\nomeqref {19}
    \nompageref{16}

\end{thenomenclature}






\end{document}